\documentclass{tlp}

\usepackage{latexsym}
\usepackage{proof}

\newtheorem{definition}{Definition}

\newtheorem{observation}{Observation} [section]
\newtheorem{theorem}{Theorem}[section]

\newtheorem{example}{Example}[section] 
\newtheorem{remark}{Remark}[section] 
\newtheorem{proposition}{Proposition}[section] 

\newcommand{\ra}{\rightarrow} 

\newcommand{\callingp}[1]{\callingpn(#1)}
\newcommand{\runningp}[1]{\runningpn(#1)}
\newcommand{\donep}[1]{\donepn(#1)}
\newcommand{\callingpn}{\texttt{act}}
\newcommand{\runningpn}{\texttt{run}}
\newcommand{\donepn}{{\texttt{end}}}

\newcommand{\dropp}{\texttt{drop}}
\newcommand{\wrapp}[1]{\texttt{wrap}(#1)}

\newcommand{\consumep}{\texttt{consume}}
\newcommand{\gainp}{\texttt{gain}}
\newcommand{\fieldp}{\texttt{field}}
\newcommand{\splitp}[2]{\texttt{split}\langle#1\rangle\{#2\}}

\newcommand{\aeminium}{\AE\-mi\-nium}
\newcommand{\To}{\Rightarrow}

\newcommand{\ot}{\otimes}
\newcommand{\voidp}{\texttt{void}}
\newcommand{\thisp}{\texttt{this}}
\newcommand{\pe}[1]{\texttt{end}(#1)}

\newcommand{\syncp}[1]{\texttt{sync}(#1)}
\newcommand{\countp}[1]{\texttt{ct}(#1)}
\newcommand{\dgp}[1]{\texttt{dg}(#1)}

\newcommand{\pr}[1]{\texttt{ref}(#1)}
\newcommand{\gparam}[1]{\texttt{gparam}(#1)}

\newcommand{\pu}[1]{\texttt{unq}(#1)}
\newcommand{\perset}{\texttt{PER}}
\newcommand{\gperset}{\texttt{GPER}}

\newcommand{\restenv}{\texttt{r\_env}}
\newcommand{\unqc}{\texttt{unq}}
\newcommand{\zeroc}{{\bf 0}}
\newcommand{\shrc}{\texttt{shr}}
\newcommand{\nogroup}{\texttt{ndg}}
\newcommand{\nostatement}{\texttt{nst}}
\newcommand{\envp}[1]{\texttt{env}(#1)}
\newcommand{\newgroup}[1]{\texttt{group}\langle#1 \rangle}
\newcommand{\okc}{\texttt{ok}}
\newcommand{\immc}{\texttt{imm}}
\newcommand{\nonec}{\texttt{none}}
\newcommand{\atomicc}{\texttt{atm}}
\newcommand{\concc}{\texttt{conc}}

\newcommand{\pn}[1]{\texttt{none}(#1)}

\newcommand{\localvar}[2]{\texttt{let}\ #1\ \texttt{in}\ #2 \ \texttt{end}}

\def\cD{\mathcal{D}}

\newcommand{\nonep}{\texttt{none}}

\newcommand\zero{0}
\newcommand\one{1}

\newcommand\bang{\mathop{!}}

\newcommand\lolli{\multimap}
\newcommand\with{\binampersand}

\newcommand{\init}{\mathsf{init}}
\newcommand{\prom}{\mathsf{prom}}

\newcommand\lra{\longrightarrow}

\newcommand\thetool{Alcove}

\long\def\comment[#1]{}
\newcommand{\tabs}{\ \ \ \  }

\newcommand{\vx}{\widetilde{x}}
\newcommand{\vy}{\widetilde{y}}
\newcommand{\vt}{\widetilde{t}}

\newcommand{\fv}{{\it fv}}

\newcommand{\absp}[3]{\forall #1 (#2 \to #3)}

\newcommand{\localp}[2]{\exists #1(#2)}

\newcommand{\bangp}[1]{ !\,#1}

\newcommand{\conf}[3]{\langle #1; #2;#3\rangle}

\newcommand{\redirex}{\longrightarrow^{*}}

\newcommand{\rTell}{\rm R_{TELL}}

\newcommand{\rChoice}{\rm R_{CHOICE}}

\newcommand{\rLocal}{\rm R_{LOC}}

\newcommand{\rCall}{\rm R_{CALL}}

\newcommand{\defsymbol}{\stackrel{\textup{\texttt{def}}}  {=}}

\newcommand{\defsymboldelta}{\stackrel{\Delta}  {=}}

\newcommand{\ccp}{\texttt{ccp}}
\newcommand{\lcc}{\texttt{lcc}}

\newcommand{\os}{[\![}
\newcommand{\cs}{]\!]}

\newcommand{\nilp}{\it{nil}}

\newcommand{\lr}[1]{\langle #1 \rangle}

\newcommand{\entails}{\vdash}

\newcommand{\notentails}{\not \entails}

\long\def\comment#1{}

\newcommand{\redi}{\longrightarrow }

\newcommand{\true}{\one}

\def\cC{\mathcal{C}}
\def\cD{\mathcal{D}}

\def\cL{\mathcal{L}}

\def\cV{\mathcal{V}}

\def\cS{\mathcal{S}}

\newcommand{\red}[1]{#1}

\newcommand{\comments}[1]{ }

\newcommand{\hide}[1]{}

\newcommand{\defrule}[2]{
  \renewcommand{\theequation}{${\rm #1}$}
  \begin{equation}
  \scriptsize{
  \boxed{#2}
  }
  \end{equation}
}
\usepackage{times}
\usepackage{subcaption}
\usepackage{hyperref}
\usepackage{stmaryrd}
\usepackage{graphicx}
\usepackage{color}
\usepackage{paralist}
\usepackage{latexsym}
\usepackage{amsmath}
\usepackage{amssymb}
\usepackage{graphics}
\usepackage{proof} 
\usepackage{todonotes} 
\usepackage{listings}
\usepackage{lineno}

\usepackage{url}
\usepackage{galois}
\usepackage{verbatim}

\usepackage{paralist}
\usepackage{balance}
\usepackage{multirow}

\usepackage{tikz}
\usetikzlibrary{arrows}

\begin{document}
\long\def\comment#1{}

\lstset{language=Java,
          morekeywords={unq,none,imm,shr,let,in,end,attr,attrs,atomic,ng,then,ask,ok,group},
          basicstyle=\ttfamily\scriptsize,breaklines=true,
          numbers=left, numberstyle=\tiny,numbersep=5pt,
          xleftmargin=20pt}


\title{A Concurrent Constraint Programming Interpretation of    Access Permissions}

\author[Olarte, Pimentel and Rueda]{CARLOS OLARTE\\
ECT, Universidade Federal Rio Grande do Norte, Brasil\\
E-mail: carlos.olarte@gmail.com
\and 
ELAINE PIMENTEL \\
DMAT, Universidade Federal Rio Grande do Norte, Brasil\\
E-mail: elaine.pimentel@gmail.com
\and
CAMILO RUEDA\\
DECC, Pontificia Universidad Javeriana Cali, Colombia\\
E-mail: camilo.rueda@gmail.com
}

\pagerange{\pageref{firstpage}--\pageref{lastpage}}
\volume{\textbf{10} (3):}
\jdate{June 2012}
\setcounter{page}{1}
\pubyear{2012}

\maketitle

\begin{abstract}

A recent trend in object oriented (OO) programming languages is the
use of Access Permissions (APs) as an abstraction for controlling concurrent
executions of programs. The use of AP source code annotations defines
a protocol specifying how object references can access the mutable
state of objects. Although the use of APs simplifies the task of
writing concurrent code, an unsystematic use of them can lead to
subtle problems. This paper presents a declarative interpretation of
APs as Linear Concurrent Constraint Programs 
(\lcc). We represent APs as constraints (i.e., formulas in
logic) in an underlying constraint system whose entailment relation
models the transformation rules of APs. Moreover, we use processes in \lcc\
to model the dependencies imposed by APs, thus allowing the
faithful representation of their flow in the program. We verify
relevant properties about AP programs by taking advantage of the 
interpretation of \lcc\ processes as formulas in Girard's intuitionistic
linear logic (ILL). Properties include deadlock detection, program
correctness (whether programs adhere to their AP specifications or not),
and the ability of methods to run concurrently. 
By relying on a focusing discipline for ILL, we  provide a complexity  measure for proofs of the above mentioned properties. 
The effectiveness of our verification techniques
is demonstrated by implementing  the Alcove tool that includes an
animator and a verifier. The former executes the  \lcc\  model, observing the flow of
APs and quickly finding inconsistencies of the APs vis-{\`a}-vis the
implementation. The latter is an automatic theorem prover based on
ILL. This paper is under consideration for publication in Theory and Practice of Logic Programming (TPLP).


 \end{abstract}

\section{Introduction}
Reasoning about concurrent programs is much harder than reasoning
about sequential ones. Programmers often find themselves overwhelmed
by the many subtle cases of thread interactions they must be aware of
to decide whether a concurrent program is correct or not. In order to
 ensure
program reliability, the programmer  needs also  to  figure out the right level of thread atomicity  to avoid race conditions, cope with mutual
exclusion requirements  and guarantee deadlock freeness. 

All these problems are  aggravated when software designers write programs
using an object oriented (OO) language and use OO strategies to design
their programs.
In an OO program, objects can have multiple references (called
aliases) that can modify local content concurrently. This
significantly increases the complexity of the design of sound
concurrent programs. For instance, 
%
 data race conditions arise when two object references read and write
concurrently from/to an object memory location.
To cope with data races, one could simply place each object access
within an atomic block, but this would affect negatively program
performance.
A better strategy could be to lock just the objects
that are shared among threads. However, it then becomes hard to
estimate which objects should be shared and which locations should be
protected by locks  just by looking at the program text.

Languages like \aeminium\ \cite{CBD:2009},  Plaid
\cite{DBLP:conf/oopsla/SunshineNSAT11} and Mezzo \cite{Pottier:2013:PPM:2544174.2500598} propose a strategy to design
sound and reliable concurrent programs based on the concept of
\emph{access permissions} (APs) \cite{Boyland:2001}. 
APs are abstractions about the aliased access to an object content and they are  annotated in the source code. 
They permit a direct control of the access to the mutable state of an
object. Making explicit the access to a shared mutable state
facilitates verification and enables parallelization of code. For
instance, a \emph{unique} AP, describing the case when only one
reference to a given object exists, enforces absence of interference
and simplifies verification;  a \emph{shared} AP,
describing the case when an object may be accessed and modified by
multiple references, allows for concurrent executions but makes
verification
trickier. 

Although  APs greatly help   to
devise static strategies for correct concurrent sharing of objects,
the interactions resulting from dynamic bindings (e.g., aliasing of variables) might still lead to
subtle difficulties. Indeed, it may happen that apparently correct
permission assignments in simple programs lead to deadlocks.

We propose a Linear Concurrent Constraint Programming
(\lcc)~\cite{fages01linear} approach for the verification of AP
annotated programs. Concurrent Constraint Programming (\ccp)
\cite{semantic-ccp,cp-book} is a simple model
 for concurrency that extends and subsumes both Concurrent Logic Programming and Constraint Logic Programming. 
Agents in \ccp\ interact by \emph{telling} constraints (i.e., formulas in
logic) into a shared \emph{store} of partial information and
synchronize by \emph{asking} if a given information can be deduced
from the store. In \lcc, constraints are formulas in Girard's
intuitionistic linear logic (ILL) \cite{Girard87} and \emph{ask}
agents are allowed to \emph{consume} tokens of information from the
store.  

We interpret AP programs as \lcc\ agents that \emph{produce}
and \emph{consume} APs when evolving. We use constraints to keep
information about APs, object references, object fields, and method
calls. Moreover, the   constraint entailment relation allows us to  verify compliance of methods
and arguments to their AP  signatures. The constraint system
specifies also how the APs can be transformed during the execution of the program. 

We are  able to verify 
AP programs by  exploiting  the declarative view of \lcc\ agents  as  formulas in ILL. 
The proposed program verification includes:
 \begin{inparaenum} 
 \item [(1)] deadlock detection;
 \item [(2)] whether it is possible for methods to be executed
   concurrently or not; and
 \item [(3)] whether  annotations adhere to the intended semantics
   associated with the flow of APs  or not.
\end{inparaenum} 

The key for this successful specification and analysis 
of AP annotations as ILL formulas
is the use of a linear logic's {\em focusing} discipline  \cite{DBLP:journals/logcom/Andreoli92}. 
In fact, by using focusing we can identify which actions need 
to interact 
with the environment or not, either for choosing a path to follow, 
or for waiting for 
a guard  to be available (e.g., the possession of an AP on a given object). 
This gives a method for
 measuring the complexity of focused ILL (ILLF) proofs in terms of  actions, hence
establishing an upper bound of the complexity for verifying the above mentioned properties. 
Moreover, as shown in~\cite{TCS16}, focusing guarantees that the  interpretation  of $\lcc$\ processes  as  ILL formulas is {\em adequate} at the highest level (full completeness of derivations): one step of computation (in \lcc)  corresponds  to one step of logical reasoning. Hence, our 
encodings of 
AP annotations as ILL formulas is faithful w.r.t. the proposed \lcc\ model. 

The contributions of  this work  are three-fold: 
\begin{inparaenum}
\item [(1)] the definition of a logical  semantics for APs, 
\item [(2)] provision of a 
  procedure for the verification of the above mentioned properties
 as well as a complexity analysis for it, and 
%
%
\item [(3)] the implementation of the verification approach as the
  \thetool\ tool  (\url{http://subsell.logic.at/alcove2/}).
\end{inparaenum} 
 The logical  structure we impose on APs thus allows us to formally  reason about the behavior of AP based programs and give a declarative account of the meaning of these annotations. 
 It is worth noticing that we are not considering a specific
AP based language. Instead, we give a logical meaning to the
machinery of APs and type states (see Section \ref{sec:con}) present in different  languages. This allows us to provide
static analyses independent of the runtime system at hand. For concreteness,
 we borrow the AP model of \aeminium~\cite{CBD:2009}, a  concurrent OO programming language  based on the idea of APs.

The paper is organized as follows. Section \ref{sec:aem} presents the
syntax of the AP based language used here and
Section~\ref{subsec:lcc-def} recalls \lcc. Section~\ref{sec:encoding}
presents the interpretation of AP programs as \lcc\ agents. We also show how the
proposed model is a runnable specification that allows observing the
flow of a program's permissions. We implemented this model as the Alcove 
LCC Animator explained in
Section~\ref{sec:tool}. Section~\ref{sec:verification} describes our
approach to program verification and its implementation as the Alcove LL
prover. It also presents a complexity analysis of the proposed
verification.  Two compelling examples  of our framework are
described in Section \ref{sec:app}: the verification of a critical
zone management system and a concurrent producer-consumer system.
Section~\ref{sec:conclusion} concludes the paper.

A preliminary short version of this paper was published in
\cite{DBLP:conf/ppdp/OlartePRC12}. The present paper gives many more
examples and explanations and provides precise technical details. In
particular, in this paper we identify the fragment of \lcc\ (and ILL)
required for the specification of AP programs and we show that this
fragment allows for efficient verification techniques.
Moreover, the
language (and analyses) considered here take into account \emph{Data
  Group Permissions} \cite{DBLP:conf/oopsla/Leino98,CBD:2009}, a
powerful abstraction that adds application level dependencies without
sacrificing concurrency (see Section \ref{sec:data-groups}). 


\label{sec:preliminaries}
\section{Access Permissions in Object Oriented Programs}\label{sec:aem}
\label{sec:aeminium}
\red{We start with an intuitive description of access permissions and data group access permissions. In Section \ref{sec:langAP}, we give a formal account of them. }

Access permissions (APs) are abstractions describing how objects are
accessed. Assume a variable $x$ that points to the object $o$. A
 \emph{unique} permission to the reference $x$ guarantees that $x$ is the sole reference to
object $o$. A \emph{shared} permission provides $x$ with reading
and modifying access to  $o$, which allows other references to
$o$ (called \emph{aliases}) to exist and to read from it or to modify it. The
\emph{immutable} permission  provides $x$ with read-only access to
$o$, and allows any other reference to  $o$ to exist and to read
from it. 
\begin{figure}
\begin{subfigure}{.75\textwidth}
\vspace{-5cm}
\begin{lstlisting}
class stats {...} // Definition of statistics 
class collection { // Collection of elements
 collection() none(this) => unq(this) {...}//constructor
 sort() unq(this) => unq(this) {...}
 print() imm(this) => imm(this) {...}
 compStats(stats s) imm(this),unq(s) => imm(this),unq(s) {...} 
 removeDuplicates() unq(this) => unq(this){...}}
main() {
  let collection c, stats s in
    c := new collection() 
    s := new stats()
    c.sort()
    c.print()
    c.compStats(s)
    c.removeDuplicates() 
  end}
\end{lstlisting}
\end{subfigure}
\begin{subfigure}[b]{.23\textwidth}
\hspace{-1.0cm}\includegraphics[scale=0.55]{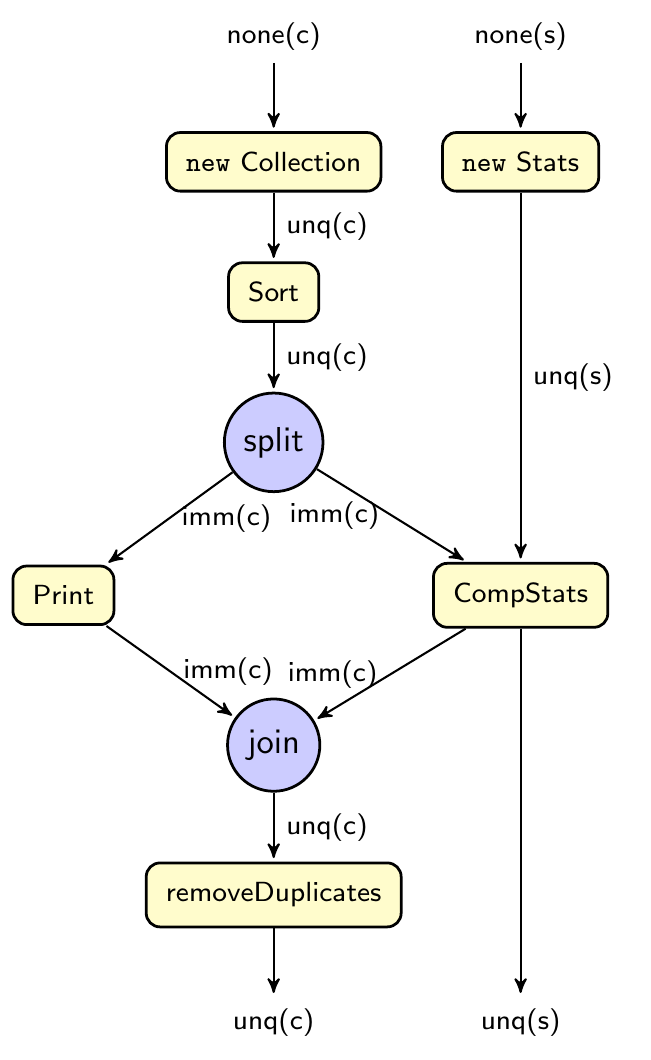}
\end{subfigure}
\caption{Example of an AP annotated program and \red{its permission flow graph.}\label{fig:ae-code}}
\end{figure} 
Let us use a simple example to explain  APs and the  \emph{concurrency-by-default} behavior \cite{CBD:2009} they offer. 
Figure \ref{fig:ae-code} shows a program, taken and slightly modified
from \cite{CBD:2009}, that operates over a collection of
elements. Starting at line 8, the program creates an object of type
$collection$ at line 10 and an object of type $stats$ at line 11. The
program sorts the collection $c$ at line 12, and prints it at line
13. It computes some statistics at line 14, and removes duplicates
from the collection at line 15. Lines 3-7 declare the signatures for
the methods. The constructor  
 builds a
unique reference to a new collection at line 3. 
Methods $sort$ and
$removeDuplicates$ modify the content of the collection and  they 
require a {unique} reference to it. Method $compStats$
requires and returns an {immutable} (read-only) AP to the collection
$c$ and a unique AP to the parameter $s$.

\red{Given the AP signature of theses methods, the AP dataflow  is
computed where: (1) conflicting accesses are ordered according to the lexical order of the program; and (2)  non-conflicting instructions  can be executed concurrently. 
For instance, methods in lines 12 and 13 cannot be executed concurrently (since  $sort$ requires a unique permission ) and methods in lines 13 and 14 can be executed concurrently (since both methods require  an immutable permission on $c$). Finally, the method in line $15$ cannot be executed concurrently with $compStats$ since $removeDuplicates$ requires a unique AP. Hence, what we observe is that the 
{unique} permission returned by the constructor is consumed by the call of method $sort$.  Once this method terminates, the
unique permission can be split into two {immutable}
permissions, and methods $print$ and $compStats$ can be executed
concurrently. Once both methods have finished, the
{immutable} APs are {joined back} into a {unique}
AP, and the method $removeDuplicates$ can be executed.}


\subsection{Share Permissions and Data Groups}
\label{sec:data-groups}
As we just showed, unique permissions can be split into several immutable APs to allow 
multiple references to \emph{read}, 
simultaneously,   the state of an object. Therefore,
from the AP annotations, the programming language can  determine, automatically,  the instructions that can be executed concurrently and those that need to be executed sequentially (Figure \ref{fig:ae-code}). 

In the  case of   \emph{share} APs, 
  several references can \emph{modify} concurrently the state of the same object. \red{Hence, the programmer needs additional  control structures to  make explicit the parts of the code that can be executed concurrently. }
 \red{Consider for instance the following excerpt of code: }
\begin{lstlisting}
let Subject s, Observer o1, Observer o2 in
  s := new Subject()
  o1 := new Observer(s) // Requires a share permission on s
  o2 := new Observer(s)
  s.update() // Requires a share permission on s
  s.update() 
\end{lstlisting}

where the  \red{constructor $Observer$ as well as the method $update$ require a share permission on $s$. 
Assume also that the intended behaviour of the program is that the method $update$  should be executed only after the instantiation of the \emph{Observer} objects.  }

\red{If we were to handle share permissions as we did with immutable permissions in the previous section, once the new instance of $Subject$ is created in line 2, the unique permission $s$ has on it can be split into several share permissions. Hence, statements in lines 3 to 6 could be executed concurrently. This means that  a possible  run of the program may execute the method $update$ in lines 5 and 6 before instantiating the $Observers$ in lines 3 and 4, which does 
not comply with the intended behavior of the program.}
%

Higher-level dependencies in AP  programs can be defined by means
of \emph{Data Groups} (DGs) \cite{DBLP:conf/oopsla/Leino98} as in
\cite{CBD:2009,DBLP:journals/toplas/StorkNSMFMA14}. Intuitively, a DG represents a collection of objects
and it controls the flow of share permissions on them.  For that, two
kinds of \emph{Data Group Access Permissions} (DGAPs) are defined: an
\emph{atomic} permission  provides exclusive access to the
DG, much like a unique AP for objects. Then, working on an atomic DGAP
leads to the sequentialization of the code; on the contrary, a 
\emph{concurrent} DGAP 
 allows other DGAPs to coexist on the same DG. Therefore,
a concurrent DGAP  allows for the parallel execution of the code.

\begin{figure}
\hspace{-1.8cm}
\begin{subfigure}{.62\textwidth}
\vspace{-3cm}
\begin{lstlisting}
class Subject <dg>{
   Subject() none(this) => unq(this) 
   update() shr : dg(this) => shr : dg(this)}
class Observer<dg>{
   Observer(Subject<dg> s) none(this),shr:dg(s)=> unq(this),shr:dg(s)}
main(){
  group<g>
  let Subject s, Observer o1, Observer o2 in
   split(g){
      s := new Subject<g>()
      o1 := new Observer<g>(s) 
      o2 := new Observer<g>(s) }
   split(g){ 
      s.update() 
      s.update()}
  }
\end{lstlisting}
\end{subfigure}
\begin{subfigure}{.25\textwidth}
\includegraphics[scale=0.55]{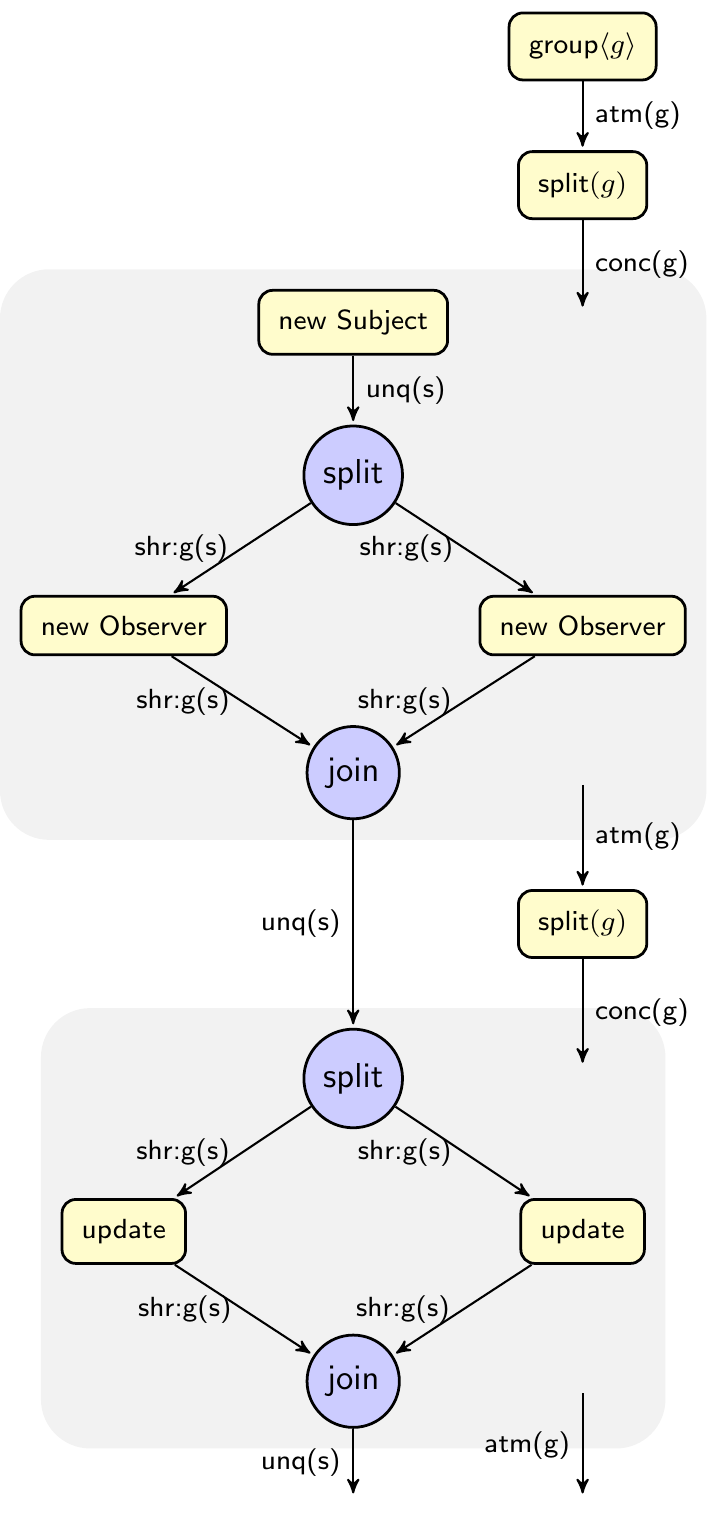}
\end{subfigure}
\caption{Subjects and Observers with data groups. \label{fig:s-o-DG}}
\end{figure}

Consider the code in Figure
\ref{fig:s-o-DG}, taken and slightly modified (syntactically) from \cite{CBD:2009}.
 In
line 1, the class $Subject$ is declared with a DG parameter $dg$. This
 parameter is similar to a type parameter (template) in modern
OO programming languages.  The method $update$ (line 3)
requires a share permission on the DG $dg$ to be invoked. The class
$Observer$ (line 4) is also defined with a DG parameter $dg$ and its
constructor requires an object of type $Subject\lr{dg}$. 

The statement $\newgroup{g}$ (line 7)
creates a DG and assigns to it an atomic DGAP. The
$\texttt{split}(g)$ instruction 
\emph{consumes}  such atomic  permission on $g$ and splits it into
three concurrent DGAP, one per each statement in the block (lines
10,11 and 12). According to the data dependencies, statements in lines 11 and
12 can be executed concurrently once $s$ is instantiated in line 10.
When the instructions in lines 11 and 12 have terminated, the  
concurrent DGAPs on $g$ are joined back to an atomic  permission on
$g$. When this happens, 
the $\texttt{split}$
block in line 13 consumes such permission and splits it into  two concurrent DGAPs so that   statements
in lines 14 and 15 can be executed concurrently. Hence, DGAPs enforce  a   \emph{data group dependency} between the two
blocks (in gray in the figure) of code.


\subsection{An Access Permission Based Language}\label{sec:langAP}
This section formalizes the syntax we have used in the previous
examples.  Our language is based on the core calculus $\mu$\aeminium
~\cite{DBLP:journals/toplas/StorkNSMFMA14} for \aeminium\ \cite{CBD:2009}, an
\red{ OO programming language where concurrent behavior arises when methods require non-conflicting APs, as exemplified in the last section}. As a core calculus, it focuses on the
mechanisms to control the flow of  permissions and it abstracts away from other implementation details of the language (e.g., control structures). 
\red{Unlike $\mu$\aeminium, our core calculus  abstracts away also  from:  (1) specific implementation details  to guarantee atomicity (e.g., the \texttt{inatomic} statement  used in the intermediate code of \aeminium\  to keep track of entered atomic blocks), (2)   mechanisms to define internal data groups, i.e., DGs in our language  can only be  declared in the main function, and (3), we do not consider class inheritance.  }

\begin{figure}
{
$
\footnotesize
\begin{array}{lrcl}
\mathrm{(group \ decl.)} & G & ::= & \newgroup{ g_1,\cdots,g_n }\\
\mathrm{(permissions)} & p & ::= & \texttt{unq}  \ \vert \ \texttt{shr}:g  \ \vert \ \texttt{imm} \  \vert \ \texttt{none}  \\
\mathrm{(programs)} & P & ::= & \langle \widetilde{CL} \  main \rangle\\
\mathrm{(class \ decl.)} & CL & ::= & \mathtt{class} \ c\ \langle \widetilde{g} \rangle\  \{\   F  \  \widetilde{\it CTR}  \  \widetilde{M} \ \} \\ 
\mathrm{(types)} & T & ::= & c \langle\widetilde{g}\rangle\\
\mathrm{(field \ decl.)} & F & ::= & \cdot \mid \mathtt{attr}\  \widetilde{T \ a}\\
\mathrm{(method \ decl.)} & {\it CTR} & ::= &  \ c(\widetilde{T \ x})\ \pn{\thisp}, \widetilde{p(x)} \ \ \To \ \  p(\thisp), \widetilde{p(x)} \ \{ s\}\\
 & M & ::= &  \ m(\widetilde{T \ x})\ p(\thisp),\widetilde{p(x)}\ \ \To \ \ p(\thisp) ,\widetilde{p(x)} \ \{ s \}\\

\mathrm{(main)} & main & ::= & \mathtt{main}   \ \{ 
G \ s \}\\
\mathrm{(references)} & r & ::= &  x \ | \ x.a \ | \ \thisp \mid \thisp.a\\
\mathrm{(righthand\ side)} & rhs & ::= &  r \mid null\\

\mathrm{(statements)} & s & ::= & \localvar{\widetilde{T \ x}}{s} \  \vert \  r\langle g  \rangle := rhs \  \vert \ \ r.m  (\widetilde{r}) \  \vert \ \\
& & &  r := \mathtt{new} \ c \langle \widetilde{g}\rangle (\widetilde{r}) \ | \  \splitp{\widetilde{g}}{\widetilde{s}} \ | \   \{\widetilde{s}\}\\
\end{array}
$
}
\caption{Syntax of AP annotated programs. 
 $c$, $m$, $a$, $x$, $g$ ranges, respectively,  over name of classes, methods, fields, variables and data groups. $\widetilde{x}$  denotes a possibly empty sequence of variables $x_1,....,x_n$. This notation is similarly used for other syntactic categories.
\label{fig-syntax}}
\end{figure}
 Programs are built from the   syntax in Fig. \ref{fig-syntax}. 
\red{DGs are declared with $\newgroup{ g_1,\cdots,g_n }$. APs on objects can be unique,  immutable or none (for null references), and  share on the  DG $g$ ($\shrc:g$)}. 
A program consists of a series of class definitions ($\widetilde{CL}$)
and a \emph{main} body. A class can be parametrized in zero or more DG names
($\lr{\widetilde{g}}$) and it contains zero or several fields
($F$), \red{constructors ($\widetilde{\it CTR})$} and methods ($\widetilde{M}$).

A class field (attribute) is declared by using a valid type and a field
name. Note that a type is simply an identifier of a class with its
respective DG parameters.
  
Constructors and methods specify the required permissions for the caller ($p(\thisp)$)
and the arguments ($\widetilde{p(x)}$) as well as the permissions
restored to the environment.

\red{A reference $r$ can be a variable, the self reference $\thisp$ or a  field selection as in $x.a$}.   As for the statements, we have the following:
\begin{itemize}[-]

\item The $\texttt{let}$ constructor allows us to create local variables. 

\item After the assignment $r\langle g\rangle := rhs$,  both $r$ and $rhs$ 
point to the same object (or null) as follows: (1) if $rhs$ has a unique permission on an object $o$, then after the assignment, $r$ and $rhs$
have a share permission $\shrc\!:\!g$ on $o$. This explains the need of ``$\langle g \rangle$'' in the syntax. As syntactic sugar one could decree that 
$r_l:=r_r$ means $r_l\lr{\texttt{default}}:=r_r$ where 
$\texttt{default}$ is a predefined DG; (2) if $rhs$ has a $\shrc\!:\!g$ (resp $\immc$) permission on $o$, then $r$ and $rhs$ end with a $\shrc\!:\!g$ (resp $\immc$) permission on $o$. Finally, (4) if $rhs=null$ or it is a null reference, then $r$ and $rhs$ end with a $\nonec$ permission. 
We note that in assignments, the  right and left hand sides are references. We do not lose generality  since it is possible to unfold more general expressions by using local variables.

\item An object's method can be invoked by using a reference to it with the right number of parameters as in $r.m(\widetilde{r})$. We assume that in a call to a method (or constructor), the
actual parameters are \emph{references} (i.e., variables, including
$\texttt{this}$, and attributes) and not arbitrary expressions. Since
we have parameters by reference, we assume that the returned type is
$\voidp$ and we omit it in the signature.

\item A new instance of a given class is created by   $r := \texttt{new}\ c\langle \widetilde{g} \rangle(\widetilde{r})$ where we specify the required DG parameters $\widetilde{g}$ and actual parameters ($\widetilde{r}$) of the constructor. 

\item For each $g_i$,  $\splitp{g_1,...,g_n}{s_1 \cdots s_m}$  consumes an atomic or a concurrent DGAP on $g_i$. Then, it splits each of such DGAPs into $m$ concurrent DGAPs (one per each statement in the block). Once the statements in the block have finished their execution, the  concurrent DGAPs created are consumed and the original  DGAPs are restored.

\item Finally, we can compose several statements in a block $\{s_1,..,s_n\}$ where the concurrent execution of statements is allowed according to the data dependencies imposed by the APs and the DGAPs:
once $s_i$ has successfully consumed its needed permissions, the execution of $s_{i+1}$ may start concurrently. Moreover, if $s_i$ cannot   acquire the needed permissions, it must wait until such permissions are released by the preceding statements. 
\end{itemize}

\begin{remark}[Circular Recursive Definitions]\label{rem-rec-def}
\red{The AP language in Fig. \ref{fig-syntax} allows us to write recursive  methods. However, the language lacks  control structures (e.g. \emph{if-then-else} statements)  to specify base cases in recursive definitions. 
This language must be understood as a core language to specify the AP mechanisms and not as  a complete OO programming language implementing the usual data and control structures. Hence, we shall assume that  there are no circular recursive definitions in the source program.
Due to the lack of control structures, this is an unavoidable restriction to guarantee  termination of the analyzes in Sections \ref{sec:encoding}  
and \ref{sec:verif}. }
\end{remark}

\paragraph{Dependencies and Execution.} Recall that blocks of sentences are enclosed by curly brackets. Hence, we say that a sentence $s$ occurs in a block if 
$s$ is inside the braces of that block. 
\begin{definition}[Conflicts and Dependencies]\label{def:conf1}
Let $s_i$ and $s_j$ be statements that occur in the same block. 
We say that $s_i$ and $s_j$ are in \emph{conflict} if both statements use an object $o$ in conflicting modes, i.e., either (1) $s_{i}$ or $s_j$ require a unique  permission on $o$, (2) $s_i$ requires a share permission on $o$
and $s_j$ requires an immutable permission on $o$, 
 or (3) $s_i$ and $s_j$ require a share permission on different data groups. 
Two blocks $ \splitp{\widetilde{g}}{s_1,...,s_n}$ 
and $\splitp{\widetilde{g'}}{s'_1,...,s'_m}$ occurring in the same block are in conflict if $\widetilde{g} \cap \widetilde{g'} \neq \emptyset$. 
 \end{definition}

\red{The semantics of AP programs share with other semantics for OO languages (see e.g., \cite{DBLP:journals/toplas/IgarashiPW01})
the rules and contexts to keep track of  references, objects
as well as lookup tables to identify class names, fields and methods with their respective definitions. 
Additionally, in the case of AP programs, the semantics relies on an evaluation context to keep track of  the DGs created as well as the available APs in the system. This context plays an important role in the semantic rules: a statement $s$ is executed only if the evaluation context possesses all the permissions required by $s$. Moreover, in order to allow  parallel executions, once $s$ consumes the needed permissions, the next statement (in the lexical order of the program) is enabled for execution. 
For instance,   in a block of statements $s_1,...,s_n$, the execution starts by  enabling  $s_1$. Each enabled statement $s_i$ checks whether the required permissions are available. If this is the case, such permissions are consumed, $s_i$ starts its execution and the next statement $s_{i+1}$ is enabled. When $s_i$ terminates its execution,  the consumed permissions are restored to the environment. On the other side, if the permissions required by  $s_i$ are not available, $s_i$ must wait   until   the needed permissions are released/produced by the preceding statements. Hence,   non-conflicting blocks lead  to concurrent executions
and conflicting blocks are sequentialized according to the lexical order of the program  (as in  Figures \ref{fig:ae-code} and \ref{fig:s-o-DG}). }

\red{The reader may refer to  \cite[Sections 3.2 and 3.3]{DBLP:journals/toplas/StorkNSMFMA14} for the semantic rules of $\mu$\aeminium\ that can be easily adapted to the sublanguage in Figure \ref{fig-syntax}. 
In Section \ref{sec:encoding} we  give a precise definition of the needed evaluation contexts by  using constraints (i.e., formulas in logic) and the state transformation by means of concurrent processes consuming and producing those constraints. }

\section{Linear Concurrent Constraint Programming}
\label{subsec:lcc-def}

Concurrent Constraint Programming (\ccp)~\cite{DBLP:conf/popl/SaraswatR90,semantic-ccp,cp-book} (see a survey in \cite{olarte-constraints}) is a
model for concurrency that combines the traditional operational view
of process calculi with a {declarative} view based on
logic. This allows \ccp\ to benefit from the large set of
reasoning techniques of both process calculi and logic. 

Agents
 in \ccp\ \emph{interact} with each other by \emph{telling}
and \emph{asking} information represented as \emph{constraints} to a
global store. The type of constraints is parametric in a
\emph{constraint system} \cite{semantic-ccp} that specifies the basic constraints  that agents can tell and ask during execution. Such systems can be specified as Scott information systems  as in  \cite{DBLP:conf/popl/SaraswatR90,semantic-ccp} or they can be specified in a suitable fragment of logic  (see e.g., \cite{fages01linear,NPV02}).

The basic constructs (processes) in \ccp\ are:  (1) the
\emph{tell} agent $c$, which adds the constraint $c$ to the store,
thus making it available to the other processes. Once a constraint is added, it cannot be removed from the store (i.e., the store grows
monotonically).  And (2), the \emph{ask}
process $c\to P$, that queries if $c$ can be deduced from the
information in the current store; if so, the agent behaves like $P$,
otherwise, it remains blocked until more information is added to the
store. In this way, ask processes define a simple and powerful synchronization mechanism
based on entailment of constraints. 

Linear Concurrent Constraint Programming (\lcc)~\cite{fages01linear} is a
\ccp-based calculus that considers constraint systems built from a
fragment of Girard's intuitionistic linear logic
(ILL)~\cite{Girard87}. The move to a \emph{linear discipline} permits
ask agents to \emph{consume} information (i.e., constraints) from the
{store}. 

\begin{definition}[Linear Constraint Systems \cite{fages01linear}]\label{def:csystem}
A linear constraint system is a pair $(\cC, \entails)$ where $\cC$ is a set of formulas (linear constraints) built from a signature $\Sigma$ (a set of function and relation symbols), a denumerable set of variables $\cV$ and the following ILL operators: 
\emph{multiplicative} conjunction ($\ot$) and its neutral element
($\one$), the existential  quantifier ($\exists$) and the exponential bang  ($\bang$). 
We shall use $c,c',d,d'$, etc, to denote elements of $\cC$. Moreover, let $\Delta$ be a set of non-logical axioms of the form $\forall \vx. [c \lolli c']$ where all free variables in $c$ and $c'$ are in $\vx$. We say that $d$ \emph{entails} $c$, written as $d \entails c$, iff the sequent $\bang\Delta, d \lra c$ is provable in ILL (Figure \ref{fig:rules-ill}).
\end{definition}

\begin{figure}
\resizebox{\textwidth}{!}{
$
\begin{array}{cccccccc}
\infer[\init]{c \lra c}{}
&\ 
\infer[\top_R]{\Gamma \lra \top}{}
&\  
\infer[1_L]{\Gamma,1 \lra c}{\Gamma \lra c}
&\  
\infer[1_R]{\lra 1}{}
 \\\\\
\infer[\otimes_L]{\Gamma, c_1 \otimes c_2 \lra c}{\Gamma,c_1,c_2 \lra c} 
&\  
\infer[\otimes_R]{\Gamma_1,\Gamma_2 \lra c_1 \otimes c_2}{\Gamma_1 \lra c_1 
\  \  \Gamma_2 \lra c_2}
&\  
\infer[\with_{Li}]{\Gamma, c_1 \with c_2 \lra c}{\Gamma,c_i \lra c} 
&\  
\infer[\with_R]{\Gamma \lra c_1 \with c_2}{\Gamma \lra c_1 \  \  \Gamma\lra c_2}
 \\\\\
\infer[\lolli_L]{\Gamma_1, \Gamma_2, c_1\lolli c_2 \lra c}{\Gamma_1 \lra c_1\  \  \Gamma_2,c_2 \lra c} 
&\  
\infer[\lolli_R]{\Gamma \lra c_1 \lolli c_2}{\Gamma, c_1 \lra c_2 }
&\  
\infer[\oplus_L]{\Gamma, c_1 \oplus c_2 \lra c}{\Gamma,c_1 \lra c
\  \  \Gamma,c_2\lra c} 
&\  
\infer[\oplus_{Ri}]{\Gamma \lra c_1 \oplus c_2}{\Gamma \lra c_i}
\\\\
\infer[\exists_L]{\Gamma,\exists x .c \lra d}{\Gamma,c \lra d \  \ x\notin fv(\Gamma,d)}
&\  
\infer[\exists_R]{\Gamma \lra \exists x.c}{\Gamma \lra c[t/x]}
&\  
\infer[\forall_L]{\Gamma,\forall x .c \lra d}{\Gamma,c[t/x] \lra d}
&\  
\infer[\forall_R]{\Gamma \lra \forall x.c}{\Gamma \lra c \  \ x\notin fv(\Gamma)}
\\\\
\infer[W]{\Gamma,\bangp c \lra d}{\Gamma \lra d}
&\  
\infer[C]{\Gamma,\bangp c \lra d}{\Gamma,\bangp c,\bangp c \lra d}
&\  
\infer[D]{\Gamma,\bangp c \lra d}{\Gamma,c \lra d}
&\  
\infer[\prom]{\bangp \Gamma \lra \bangp d}{\bangp \Gamma \lra d}\\
\end{array}
$
}
\caption{Rules for 
  Intuitionistic Linear Logic (ILL). $fv(c)$ (resp. $fv(\Gamma)$) denotes the set of
  free variables of formula $c$ (resp. multiset $\Gamma$). $\Gamma,\Delta$ denote multisets of formulas. \label{fig:rules-ill}}
\end{figure}
The connective $\otimes$ allows us to conjoin information in the store and $\one$ denotes the empty store. As usual, existential quantification is used to hide information. The exponential  $\bang c$ represents the arbitrary duplication
 of the resource $c$. The entailment  $d \entails c$ means that the information $c$ can be deduced from the information represented by  constraint $d$, \red{possibly using the axioms in the theory $\Delta$. This theory  gives meaning to (uninterpreted) predicates. For instance, if $R$ is a transitive relation, $\Delta$ may contain the axiom $\forall x,y,z. [R(x,y) \otimes R(y,z) \lolli R(x,z)]$. }
 
We assume that  ``$\bangp$'' has a tighter binding than $\otimes$  and so, we understand $\bangp c_1 \otimes c_2$ as $(\bangp c_1) \otimes c_2$. For the rest of the operators we shall explicitly use parenthesis to avoid ambiguities. \red{Given a finite set of indexes $I=\{1,...,n\}$,  we shall use $\bigotimes\limits_{i\in I}F_i$ to denote the formula $F_1 \otimes \cdots \otimes F_n$. }

We note that, according to Definition \ref{def:csystem}, constraints are built from the ILL fragment  $\otimes$, $1$, $\exists$ ,$\bang$. We decided to include all ILL connectives in Figure \ref{fig:rules-ill} (linear implication $\lolli$, 
additive  conjunction $\with$ and disjunction $\oplus$, \red{the universal quantifier $\forall$} and the unit $\top$)
since those connectives will be used to encode \lcc\ processes in Section \ref{sec:verif}. 


\subsection{The Language of Processes} \label{sec:processes}
Similar to other \ccp-based calculi, \lcc, in addition to tell and ask agents, provides constructs for parallel composition, hiding of variables, non-deterministic choices and process definitions and calls. More precisely:  

\begin{definition}[\lcc\ agents \cite{fages01linear}]\label{def:syntax-lcc} 
Agents in \lcc\ are built from constraints as follows: 
$$
P,Q,... ::= c \ | \  \red{\sum_{i\in I}\absp{\vx_i}{c_i}{P_i}}  \ | \ P \parallel Q  \ |  \   \exists \widetilde{x} (P)    \ | \ p(\widetilde{x})
$$
A \lcc\ program takes the form $\cD.P$ where $\cD$  is a set of process definitions  of the form $
p(\widetilde{y}) \defsymboldelta P
$ 
where all free variables of $P$ are in the set of pairwise distinct
variables $\widetilde{y}$. We assume $\cD$ to have a unique process definition for every process name. 
\end{definition}

Let us give some intuitions about the above constructs. The tell agent $c$  adds  constraint $c$ to the current store $d$ producing the new store $d\ot c$.

Consider   the  guarded choice
$Q= \sum\limits_{i\in I} \forall \widetilde{x_i} (c_i \to
  P_i)$ where $I$ is a finite set of indexes. 
 Let $j\in I$, $d$ be the current store and  
$\theta$ be the substitution $[\widetilde{t}/\widetilde{x_j}]$ \red{where $\widetilde{t}$ is a sequence of terms}.
 If  $d\entails d' \ot c_j\theta$ for some $d'$, then $Q$ 
  evolves into $P_j[\widetilde{t}/\widetilde{x_j}]$ and consumes
  $c_j\theta$. 
  If none of the guards   $c_i$  can be deduced
from $d$, the  process $Q$ blocks until more information is added to the
store. \red{Moreover, if many guards can be deduced, one of the alternatives is non-deterministically chosen for execution.  }
To simplify the notation, 
we shall omit ``$\sum\limits_{i\in I}$''
  in $\sum\limits_{i\in I}\forall
\widetilde{x_i}(c_i\to P_i)$
   when $I$ is a singleton; 
if the sequence of variables $\widetilde{x}$ is empty, we shall write $c \to P$ instead of 
$\forall \widetilde{x}(c\to P)$; \red{moreover, 
if $| I| =2$, we shall use ``$+$'' instead of ``$\sum$'' as in $c_1 \to P_1 + c_2 \to P_2$. }
  \\

The \red{interleaved} parallel composition of $P$ and $Q$ is denoted by  $P \parallel Q$. We shall use $\Pi_{i\in I} P_i$ to denote the parallel composition $P_1 \parallel \cdots \parallel P_n$, where $I=\{1,2,...,n\}$. If $I=\emptyset$, then  $\Pi_{i\in I} P_i = 1$ . 

 The agent $\exists \vx(P)$ behaves like  $P$
and binds the variables $\widetilde{x}$ to be local to it.
The processes $\localp{\vx}{P}$ and $\absp{\vx}{c}{P}$, as well as the constraint $\exists \vx(c)$, 
bind the variables $\vx$  in $P$ and $c$.  We shall use   $\fv(P)$ and $\fv(c)$ to denote, respectively, the  set of free variables of  $P$ and $c$.

Finally, given  a process declaration of the form
$
p(\widetilde{y}) \defsymboldelta P
$,   $p(\widetilde{x})$ 
evolves into $P[\widetilde{x}/\widetilde{y}]$. 


\subsection{Operational Semantics}
Before giving a formal definition of the operational semantics of \lcc\ agents, let us give an example of how processes evolve. For that, we shall use 
$\langle P;c\rangle \redi \langle P';c'\rangle$ to denote that the agent
$P$ under store $c$ evolves into the agent $P'$ producing the store
$c'$. This notation will be precisely defined shortly. 

\begin{example}[Consuming Permissions]\label{ex:sos-lcc}
  Let us assume a constraint system with  predicates $\texttt{ref}/3$, $\texttt{ct}/2$; constant symbols  $\unqc$, $\immc$, $\nonec$, $0$, $\nilp$; function symbol $s$ (successor); and equipped with the axiom:
\[
\red{ \Delta = \forall x, o. [\pr{x,o,\immc} \ot \countp{o,s(0)} \lolli \pr{x,o,\unqc}\ot \countp{o,s(0)}]}
 \]  
Informally,  $\Delta$ says that an $\immc$ permission can be \emph{upgraded} to $\unqc$ if there is only one reference pointing to $o$. 

Consider now  the  processes 
\[
\scriptsize
\begin{array}{lll}
P_1 & = & \pr{x,o_x, \immc} \ot \countp{o_x,s(0)} \\
P_2 & = &\pr{y,o_y, \immc} \ot \pr{z,o_y, \immc}   \ot\countp{o_y,s(s(0))} \\
Q & = & \red{\forall o (\pr{x,o,\unqc} \otimes \countp{o,s(0)} \to Q' )} \\
Q' & = &   \red{\pr{x,\nilp,\nonec} \ot \countp{o,0}}\\
R & = & \forall o (\pr{y,o,\unqc}  \to R' )\\
\end{array}
\]
  Roughly, $P_1$ adds to the store the information required to
state that $x$ points to $o_x$ with permission $\immc$ and that there is exactly  one reference to $o_x$. $P_2$ states that there are two references ($y$ and $z$) pointing to the same object $o_y$.  Process $Q$, in order to evolve,  requires  $x$ to have a unique permission on $o_x$. \red{Finally, $R$ is asking whether $y$ has a unique permission on a given object $o$ to  execute $R'$ (not specified here). }

\red{Starting from the configuration $\langle P_1\parallel P_2 \parallel Q \parallel R ; \one \rangle$, we observe the  derivation below}:

\resizebox{\textwidth}{!}
{
$
\begin{array}{lll}
{\scriptstyle\bf (1)} & & \langle P_1\parallel P_2 \parallel Q \parallel R ; \one \rangle  \\
{\scriptstyle\bf (2)}& \redi&  \langle P_1 \parallel Q \parallel R ; \one \ot 
\pr{y,o_y, \immc} \ot \pr{z,o_y, \immc}   \ot\countp{o_y,s(s(0))}
 \rangle    \\
{\scriptstyle\bf (3)}& \redi& \langle  Q \parallel R ; \one \ot 
\pr{y,o_y, \immc} \ot \pr{z,o_y, \immc}   \ot\countp{o_y,s(s(0))} \ot \pr{x,o_x, \immc} \ot \countp{o_x,s(0)}   \\
{\scriptstyle\bf (4)}&\redi&  \langle Q'[o_x/o]  \parallel R ; \pr{y,o_y, \immc} \ot \pr{z,o_y, \immc}   \ot\countp{o_y,s(s(0))}    \rangle \\
{\scriptstyle\bf (5)}&\redi&  \langle  R ; \pr{y,o_y, \immc} \ot \pr{z,o_y, \immc}   \ot\countp{o_y,s(s(0))} \ot \pr{x,\nilp,\nonec} \ot \countp{o_x,0}    \rangle 
 \end{array}
$
}

\ \\ From the initial store $\one$, neither $Q$ nor $R$ can  deduce
their guards and they remain blocked (line 1). Tell processes  $P_1$ (line 3) and $P_2$ (line 2) evolve by adding information to the store. 
\red{Let $d$ (resp. $d'$) be the store in the configuration of line 3 (resp. line 4) and    $c$ be the guard of  the ask agent $Q$. We note that 
  $d \entails d' \otimes c [o_x / o]$. Recall that checking this entailment amounts to prove in ILL the sequent $!\Delta , d \lra d' \otimes c[o_x/o]$ (Definition \ref{def:csystem}). In this case, the axiom $\Delta$ allows us to  transform
$\pr{x,o_x,\immc} \ot \countp{o_x,s(0)}$ into $\pr{x,o_x,\unqc}\ot \countp{o_x,s(0)}$. Hence,   $Q$ reduces to the tell agent $Q'[o_x/o]$ and  consumes part of the store leading to the store  $d'$ in line 4.} In line 5, $Q'[o_x/o]$ adds more information to the store, namely,  $x$ points to $null$ and there are no references pointing to $o_x$. Note that $R$ remains blocked since the guard $\pr{y,o,\unqc}$ cannot be entailed.


\end{example}

\begin{figure}
\resizebox{.98\textwidth}{!}{
$
\begin{array}{c}
\infer[\rTell]{\conf{X}{\Gamma,c}{d}  \redi  \conf{X}{\Gamma}{d\ot c}}{
}
\qquad
\infer[\rChoice]{\conf{X}{\Gamma,\sum_{i\in I}\absp{\vx_i}{c_i}{P_i}}{d} \redi
\conf{X\cup \vy}{\Gamma,P_i[\vt/\vx]}{d'}
}{
d \entails \exists \vy (d' \otimes c_i[\vt/\vx_i]), ~~\vy \cap \fv(X,\Gamma,d) = \emptyset, \ mgc(d',\vt)
}\\\\

\infer[\rLocal]{\conf{X}{\Gamma,\localp{\widetilde{y}}{P}}{d} \redi \conf{X \cup \widetilde{y}}{\Gamma,P}{d}}{\vy \cap X = \vy \cap \fv(\Gamma,d) =   \emptyset} 
\qquad
\infer[\rCall]{\conf{X}{\Gamma,p(\widetilde{x})}{d} \redi 
\conf{X}{\Gamma,P[\widetilde{x}/\widetilde{y}]}{d}
}{p(\widetilde{y}) \defsymboldelta P \mbox{\ is a process definition}}
\end{array}
$
}
\caption{Operational semantics of \lcc. $\fv(\Gamma,d)$ means $\fv(\Gamma) \cup \fv(d)$. $\fv(X,\Gamma,d)$
means $\fv(\Gamma,d) \cup X$.  The notion of  most general choice ($mgc(d',\vt)$)  is 
 in Definition~\ref{def:mgc}.
\label{fig:sos}}
\end{figure}

\paragraph{Operational Semantics} 
Let us extend the processes-store configurations used in Example \ref{ex:sos-lcc}  to consider configurations of the form $\langle X;\Gamma;c \rangle$. Here $X$ is the set of local (\emph{hidden}) variables in $\Gamma$ and $c$, $\Gamma$ is a multiset of processes of the form $P_1,...,P_n$ representing the parallel composition $P_1 \parallel \cdots \parallel P_n$, and  $c$ represents the current store. In what follows, we shall indistinguishably use the notation of multiset as parallel composition of processes. 

The transition relation $\redi$ defined on  configurations is the least relation satisfying the rules in Figure \ref{fig:sos}.  \red{We shall use 
  $\redi^*$ to denote the reflexive and transitive closure of $\redi$. }
It is easy to see that  rules $\rTell$, $\rLocal$  and $\rCall$  realize the behavioral intuition given in the previous section. Let us explain the Rule $\rChoice$. 
Recall that the process 
$\sum\limits_{i\in I}\forall
\widetilde{x_i}(c_i\to P_i)$
executes $P_j[\vt/\vx_j]$ if $c_j[\vt/\vx_i]$ can be deduced from the current store $d$, i.e., $d\entails d' \otimes c_j[\vt/\vx]$. Moreover, the constraint $c_j[\vt/\vx]$ is consumed from $d$ leading to the new store $d'$. 
Hence,  $d'$ must be the most general choice  in the following sense:

\begin{definition}[Most general choice ($mgc$) \cite{DBLP:conf/ppdp/Martinez10,DBLP:journals/tplp/Haemmerle11}]\label{def:mgc}
Consider the entailment $d \entails \exists \vy (e \otimes c[\vt/\vx])$
and assume that \red{$\vy\cap fv(d) = \emptyset$}.   Assume also that  $d \entails \exists \vy (e' \otimes c[\vt'/\vx])$ for  an arbitrary $e'$ and $\vt'$. 
We say that 
  $e$ and $\vt$ are the most general choices, notation $mgc(e,\vt)$, whenever  $e' \entails e $ implies $e\entails e'$ and $c[\vt/\vx ] \entails c[\vt'/\vx ]$.
\end{definition}

\red{The $mgc$ requirement in rule $\rChoice$ 
prevents from  an unwanted weakening of the store. For instance, 
consider the ask agent $Q = c\to P$. We know  that  $\bang c$ entails $c \otimes 1$ (i.e., $\bang c \entails c \otimes 1$). Hence, without the $mgc$ condition,  $Q$  may consume $\bang c$  leading to the store $\one$. This is not satisfactory since $Q$ did not consume the \emph{minimal information} required to entail its guard. In this particular case, we have to consider the entailment $\bang c \entails \bang c \otimes c$ where $Q$ can  entail its guard and the store remains the same. For  further details, please refer to  \cite{DBLP:journals/tplp/Haemmerle11}. }


\paragraph{Sequential Composition}
In the subsequent sections we shall use the  derived operator $P;Q$ that delays the execution of $Q$ until \red{$P$ signals its termination.} This operator can be encoded in \lcc\ as follows. Let $z$ be a variable  that does not
occur  in $P$ nor in $Q$ and let $\syncp{\cdot}$ be an uninterpreted predicate symbol that does not occur in the program. 
The process $P;Q$  is defined as
$
\exists z (\cC\os P \cs_z \parallel \syncp{z} \to Q)
$
 where $\os \cdot \cs_z$ is  in Figure \ref{fig:seq-comp}.  Intuitively, $\cC \os P\cs_z$ adds the constraint $\syncp{z}$ \red{to signal the termination of $P$}.   Then, the ask agent $\syncp{z}\to Q$ reduces to $Q$.  Note that in a parallel composition $P\parallel R$, one has to wait for the termination of  both $P$ and $R$ before adding the constraint  $\syncp{z}$. 
\red{For that,   $\cC\os P \parallel R\cs_z$ creates  fresh variables $w_1$ and $w_2$  to
signal the termination of, respectively, $P$ and $R$. Then, 
it adds $\syncp{z}$  only when both $\syncp{w_1}$ and $\syncp{w_2}$ can be deduced}. 
Assume now a process definition of the form 
$p(\vy)\defsymbol P$. We require the process $P$ to emit the constraint $\syncp{z}$ to synchronize with the call $p(\vx)$. We then add an extra parameter to the process  definition ($
p(\vy,z)\defsymboldelta \cC \os P \cs_z
$). Hence  the variable $z$ is  passed as a parameter and used by $\cC\os P \cs_z$ to synchronize with the call $p(\vx,z)$. 

\begin{figure}
\resizebox{\textwidth}{!}{
$
\begin{array}{rll l rll}
\cC\os c\cs_z &=& c \ot \syncp{z} &\quad & 
\cC\os \sum\limits_{i\in I} \forall \vx_i(c_i \to P_i)\cs_z &=& \sum\limits_{i\in I} \cC\os \forall \vx_i(c_i \to P_i)\cs_z \\

\cC\os \forall \vy (c\to P)\cs_z &=&  \forall \vy (c\to \cC\os P\cs_z ) & \quad &
\cC\os\exists y(P) \cs_z &=& \exists y(\cC\os P\cs_z) \\

\multicolumn{7}{l}{\cC\os P_1 \parallel \ldots\parallel P_n\cs_z = \exists w_1\ldots w_n ( \cC\os P_1 \cs_{w_1} \ \parallel \  \ldots\parallel
\cC\os P_n \cs_{w_n}  \parallel \bigotimes\limits_{i\in 1..n}\syncp{w_i}\to \syncp{z})}
\\

%
%

\cC\os p(\widetilde{x}) \cs_z &=& p(\widetilde{x},z) & \quad &
\cC\os p(\vy) \defsymboldelta P \cs_z &= &
 p(\vy,z) \defsymboldelta\cC\os  P \cs_z
\end{array}
$
}
\caption{Definition of the sequential composition $P;Q$ \label{fig:seq-comp}}
\end{figure}
\section{AP Programs as LCC Processes}
\label{sec:encoding}
This section  presents an  interpretation of Access Permissions (APs) and Data Group Access Permissions (DGAPs) as processes in   \lcc. 
We thus endow AP  programs 
with a declarative semantics 
which is adequate  to verify  relevant properties as we show later. 
We start defining the constraint system 
we shall use. Constants, predicate symbols and non-logical axioms are  depicted in Figure \ref{fig:signature} and explained below. 

\red{We shall use $c$, $m$, $a$, $g$, $o$ to range, respectively,  over name of classes, methods, fields, DGs and objects in the source AP language. For variables, we shall use $x,y$ and $u$. 
 We may also use primed and subindexed version of these letters. We shall  use the same letters in our encodings. Hence, if $x$ occurs in a constraint   (see e.g., predicate $\pr{\cdot}$ below), it should be understood as the representation of a variable $x$ in the source language. Finally, we shall use $z,w$ (possible primed or subindexed)  to represent identifiers of statements in the source program. Those variables will appear in the scope of constraints used for  synchronization in the model as, e.g.,  in the constraint $\syncp{\cdot}$. 
}\\

\begin{figure}
\begin{center}
\scriptsize{
\begin{tabular}{| l | l |}
\cline{1-2}
\multicolumn{2}{|c|}{\bf Constant Symbols} \\
\cline{1-2}
$\perset \!\!=\!\{\unqc,\! \shrc, \immc, \nonec\}$& \!\!\!\!\!\!Types of access permissions. \\ \cline{1-2} 
$\gperset=\{ \atomicc, \concc\}$ &  \!\!\!\!\!\!Types of data group access permissions. \\ \cline{1-2} 
$\nogroup$ & \!\!\!\!\!\!Absence of data group. \\ \cline{1-2} 
$\nostatement$ & \!\!\!\!\!\!Absence of statement. \\ \cline{1-2} 
$\nilp$ & \!\!\!\!\!\!Null reference \\ \cline{1-2}
$c\_a$ & \!\!\!\!\!\!For each field $a$ of a class $c$.  \\ \cline{1-2}
$c\_g_i$ & \!\!\!\!\!\!For each group parameter in the class definition $\mathtt{class} \ c\ \langle {g_1,...,g_n} \rangle$  \\
 \cline{1-2}
$g_1,...,g_n$ & \!\!\!\!\!\!For each DG in $\newgroup{ g_1,\cdots,g_n }$  \\ \cline{1-2} 
 \cline{1-2}
\multicolumn{2}{|c|}{\bf Predicate Symbols} \\
\cline{1-2}
$\pr{x,o,p,g}$ & \!\!\!\!\!\!$x$ points to object $o$ with permission $p\in \perset$ and belongs to the data group $g$ \\ \cline{1-2}
$\fieldp(u,o, a)$ & \!\!\!\!\!\!$u$ is the reference to field $a$ of object $o$. \\ \cline{1-2}
$\gparam{c\_g,o, gp}$ & \!\!\!\!\!\!The group parameter $g$ of the object $o$ was instantiated with the data group $gp$. \\ \cline{1-2}
 $\syncp{z}$ & \!\!\!\!\!\!Synchronizing on variable $z$. \\ \cline{1-2} 
  $\red{\callingp{z}}$ & \!\!\!\!\!\!Activate/start statement $z$. \\ \cline{1-2} 
  $\red{\runningp{z}}$ & \!\!\!\!\!\!Statement $z$ is being executed. \\ \cline{1-2} 
  $\red{\donep{z}}$ & \!\!\!\!\!\!End of statement $z$. \\ \cline{1-2} 
 $\countp{o,n}$ & \!\!\!\!\!\!There are $n$ references pointing to object  $o$. \\ \cline{1-2}
 $\dgp{g,p,z}$ &\!\!\!\!\!\!\! Statement $z$ has a data group permission of type   $p\in\gperset$ on the data group $g$.  \\ \cline{1-2}
 \multicolumn{2}{|c|}{\bf Axioms} \\ \cline{1-2}
$\texttt{downgrade}_1$ &  \!\!\!\!\!$\red{\forall x,o,g.[\pr{x,o,\unqc,\nogroup}   \lolli \pr{x,o,\shrc,g}]}$ \\ \cline{1-2}
$\texttt{downgrade}_2$ &  \!\!\!\!\!$\red{\forall x,o. [\pr{x,o,\unqc,\nogroup}  \lolli  \pr{x,o,\immc,\nogroup}]}$ \\\cline{1-2}
  $\texttt{upgrade}_1$ &  \!\!\!\!\!\!\! $\red{\forall x, o, g.[\pr{x,o,\shrc,g} \ot \countp{o,s(\zeroc)} \lolli   \pr{x,o,\unqc,\nogroup} \ot \countp{o,s(\zeroc)}]} $   \\ \cline{1-2}
$\texttt{upgrade}_2$ & \!\!\!\!\!\!$\red{\forall x,o.[ \pr{x,o,\immc,\nogroup} \ot \countp{o,s(\zeroc)}   \lolli   \pr{x,o,\unqc,\nogroup} \ot \countp{o,s(\zeroc)} ]   }$ \\ 
\cline{1-2}
\end{tabular}
}
\end{center}
\caption{Constraint system for Access Permissions. $\zeroc$ denotes the constant ``zero'' and $s(\cdot)$ successor. 
\label{fig:signature}}
\end{figure}

\noindent{\bf Permissions and constants:} Constant symbols in sets $\perset$ and $\gperset$   represent the kind of APs and DGAPs  available in the language.
Since \emph{none}, \emph{unique} and \emph{immutable}  AP are not associated to any data group, we shall use the constant $\nogroup$ to denote ``no-group''. 
Recall that the $\texttt{split}$ command splits a DGAP  into several DGAPs, one per each statement in the block. Then, we require to specify in our model the statement to which the concurrent permission is attached to (see predicate $\dgp{\cdot}$ below).
Since atomic DGAPs  are not attached to any  particular statement in the program, we
use the constant $\nostatement$ to denote  ``no-statement''.
 The constant $\nilp$ is used to denote a null reference. 
Assume  a class $c$ with  an attribute $a$ and a data group parameter $g$. We use the constant symbol $c\_a$   to make reference to  $a$ (see predicate $\fieldp(\cdot)$  below) and a  constant symbol $c\_g$ to make reference to   $g$ (see   $\gparam{\cdot}$ below). We also consider the constant symbols $g_1,...,g_n$ to give meaning to the statement $\newgroup{ g_1,\cdots,g_n }$.
\\

\noindent{\bf References and Fields:} We  use the predicate symbol $\pr{x,o,p,g}$ to represent that the variable $x$ is pointing to object $o$ and it has a permission $p$ on it. The last parameter of this predicate  is used to give meaning to share permissions of the form $\shrc:g$. As we already explained, $g=\nogroup$ when $p\neq \shrc$. 
The predicate  $\fieldp(x_u,o, a)$ associates  the variable $x_u$ to the field $a$  of object $o$. 
Once an object of a given class $c$ with data group parameters is instantiated, the predicate $\gparam{c\_g,o,gp}$ dictates that the group parameter $g$ of the object $o$ was instantiated with the DG  $gp$.
The predicate  $\syncp{z}$ 
is  used in  the definition of  $P;Q$ as explained in the previous section. 
\red{Constraints $\callingp{z}$, $\runningp{z}$ and  $\donep{z}$ 
represent, respectively, that statement $z$ has been called,
it is currently being executed or it has finished. We shall use those constraints as witnesses for verification purposes. }
The number of references (alias) pointing to a given object are modeled with the predicate $\countp{o,n}$. 
Given a data group $g$, the predicate $\dgp{g,p,z}$ dictates that the statement $z$ has   a   data group permission $p \in \gperset$  on  $g$. If $p=\atomicc$  then $z=\nostatement$.
\\

\noindent{\bf Non-logical axioms:} The entailment of the constraint system  allows us to formalize when a given AP can be transformed into another.
Assume that $x$ has a unique permission on $o$. 
Unique permissions can be downgraded to share or immutable permissions as dictated by axioms $\texttt{downgrade}_1$ and 
$\texttt{downgrade}_2$ respectively. 
Axiom $\texttt{upgrade}_1$
(resp. $\texttt{upgrade}_2$) builds a unique permission from a share (resp. immutable) permission. For that,  $x$ needs to be the unique reference with share or immutable permission to the pointed object. Conversions from share permissions into immutable and vice versa require  to first upgrade the permission to unique and then applying the appropriate downgrade axiom. 

\subsection{Modeling Statements.}\label{subsec:enc:stm}
Given an AP annotated program, we shall build a $\lcc$\ program $\mathcal{D}.P$ where  $\mathcal{D}$ includes 
process definitions for each method and constructor of the AP program
(Section \ref{sec:defs} below) 
 and a process definition to encode assignments ($\texttt{assg} $ in Figure \ref{ref:rule-assg}).   The process $P$ represents the encoding of the main body of the AP program where each   statement $s$  is encoded as a \lcc\ process $ \cS\os s
\cs_z^{G}$  that models its behavior. 

The process $ \cS\os s
\cs_z^{G}$
  adheres to the following schema.   
The \lcc\ variable $z$ is used to represent the statement $s$ in the model. 
We assume (by renaming variables if
necessary) that $z$  does not occur in $s$. 
\red{The encoding  uses constrains to signal  three possible states in the execution of $s$. When the program control reaches the statement $s$, the encoding adds the constraint $\callingp{z}$ to signal that $s$ is ready to be executed. When the needed permissions for $s$ are successfully acquired,   $\callingp{z}$ is consumed and  constraints $\syncp{z}$ and $\runningp{z}$ are added. The first one is used to 
synchronize with the rest of the model. More precisely,
the encoding of the next instruction in the program  waits  for constraint $\syncp{z}$ to be posted before starting its execution. In this way, we model the   \emph{data dependencies} resulting from the flow of APs. Constraint $\runningp{z}$ signals that $s$ is currently being executed. Once $s$ has finished and the consumed permissions are restored, the encoding consumes $\runningp{z}$ and adds the constraint $\pe{z}$. }

The  $G$ in 
$\cS\os s\cs_z^{G}$
stands for the set of DGs on which $s$ must have a concurrent DGAP to be executed.
Recall that such permissions are assigned by a $\texttt{split}$ command. Then, we use $G$ to control which DGAPs must be consumed and restored by   $s$. 

In the following we define
$\cS\os s \cs_z^G$ for each kind of statement in the syntax in Figure
\ref{fig-syntax}. For that, the following shorthand will be useful  ($\defsymbol$ must be understood as a shorthand  and not as process definition): 

\defrule{wrap}{
{
\wrapp{P,\{g_1,\cdots,g_n\},z} \defsymbol 
\red{\callingp{z};}(\bigotimes\limits_{i\in 1..n} \dgp{g_i,\concc,z} \to \true );  P \parallel ( \pe{z} \to \bigotimes\limits_{i\in 1..n} \dgp{g_i,\concc,z}))
}
}

Assume that $s$ is an statement and   $P = \cS \os s\cs_z^G$. The process    $\wrapp{P,G,z}$  first consumes all the concurrent DGAPs available for $s$, i.e., those in the set $G$. 
If $G=\emptyset$, then $\bigotimes\limits_{i\in 1..n} \dgp{g_i,\concc,z}$ is defined as $\true$. 
Observe that  once $s$ has terminated (i.e., the constraint $ \pe{z}$ is added to the store) such permissions are restored.\\

\noindent{\bf Assignments. }
\red{We have different cases for the assignment $r\langle g\rangle:= rhs$ depending whether $r$ and $rhs$ are variables or field selections. 
Let us start with the case when} both  are variables as in  $x\lr{g}:=y$ and $x$ is syntactically different from $y$. We  have: 

\defrule{R_{ALIAS}}{
\cS\os x \langle gt \rangle:=y\cs_z^G  \ =  \wrapp{\texttt{assg}(x,y,z,gt),G,z }
}

\noindent
where $\texttt{assg}$ is defined in Figure \ref{ref:rule-assg}. 
\begin{figure} 
\resizebox{\textwidth}{!}{
$
\begin{array}{llll}
  &\texttt{assg}(x,y,z,gt)&  \defsymboldelta  &  \dropp(x) ;   \gainp(x,y,gt); \red{\callingp{z} \to \runningp{z} ; \runningp{z} \to \syncp{z} \otimes !\pe{z}}\\
& \dropp(x)&    \defsymbol    &  \forall o,n,g ((\pr{x,\nilp,\nonec,\nogroup} \to \true) \ +   \sum\limits_{p\in\perset \setminus\{\nonec\}}  \pr{x,o,p,g}\ot \countp{o,s(n)}\to \countp{o,n}) \\
 & \gainp(x,y,gt)&  \defsymbol    &\pr{y,\nilp,\nonec,\nogroup}  \to  \pr{x,\nilp,\nonec,\nogroup}  \ot \pr{y,\nilp,\nonec,\nogroup}\\

  &\multicolumn{3}{l}{ + \ \ \forall o,n((\pr{y,o,\unqc,\nogroup} \ot \countp{o,s(\zeroc)}  \to \pr{y,o,\shrc,gt} \ot   \pr{x,o,\shrc,gt} \ot \countp{o,s(s(\zeroc))} )}  \\
  &\multicolumn{3}{l}{ \tabs\tabs \  +\  (\pr{y,o,\shrc,gt} \ot \countp{o,n} \to \pr{y,o,\shrc,gt} \ot   \pr{x,o,\shrc,gt} \ot \countp{o,s(n)} )}  \\
    &\multicolumn{3}{l}{\tabs\tabs \  +\  (\pr{y,o,\immc,\nogroup} \ot \countp{o,n}  \to   \pr{y,o,\immc,\nogroup} \ot   \pr{x,o,\immc,\nogroup} \ot \countp{o,s(n)}  ))}\end{array}
$}
\caption{Auxiliary definitions for Rule {$\rm R_{ALIAS}$} \label{ref:rule-assg}}
\end{figure}
The variable $x$ loses its
permission to the pointed object $o$, and the object $o$ has one less
reference pointing to it (Definition $\dropp$). Thereafter, $x$
and $y$ point to the same object and the permission of $y$ is split
between $x$ and $y$ as explained in Section \ref{sec:langAP} (Definition $\gainp$).
Finally, once the permission to $y$ is split, the constraints
$\syncp{z}$ and $!\pe{z}$ are added to the store to, respectively,  synchronize with the rest of the
program and mark the termination of the statement. \red{Note in $\texttt{assg}$ the use of the constraints $\callingp{\cdot},\runningp{\cdot}$ and $\donep{\cdot}$. Initially,  constraint $\callingp{z}$ is added (by $\texttt{wrap}$). When the permissions on $x$ and $y$ are split (after  $\texttt{drop}$ 
and $\texttt{gain}$),  $\callingp{z}$ is consumed to produce $\runningp{z}$. Finally, $\runningp{z}$ is consumed to produce $\donep{z}$.  }

\red{Now consider the case $\cS\os x.a\langle  g \rangle:=y\cs_z^G$. }
If the variable $x$ points to the object $o$ of class $c$, then the
field $a$ of $o$ can be accessed via the variable $u$ whenever the constraint 
$\fieldp(u,o,c\_a)$ holds. Intuitively, $u$ points to $x.a$ and then, a
constraint $\pr{u,o',p,g}$ dictates that  $x.a$   points  to $o'$ with permission $p$. As we shall show later, the model of constructors adds the constraint $\bangp \fieldp(u,o,c\_a)$ to establish the connection between objects and their fields. 
The
model of the assignment $\cS\os x.a\langle  g \rangle := y\cs_z^G$ is thus obtained from that
of $\cS\os u\langle  g \rangle := y\cs_z^G$: 
\defrule{\rm R_{ALIAS_F}}{
\begin{array}{l}
\cS\os x.a\langle  g \rangle:=y\cs_z^G =  \forall u,o,p,g (  \pr{x,o,p,g}\ot \fieldp(u,o,c\_a) 
  \to   (\pr{x,o,p,g}; \cS\os u\langle  g \rangle:=y \cs_z^G))
\end{array}
}
\red{The cases 
$ x.a\langle  g  \rangle := y.a' $ and $  x\langle  g \rangle := y.a $
are similar.} \\

\noindent{\bf Let.}  Local variables in the AP program   are defined as  local  variables in \lcc:
\defrule{\rm R_{LOC}}{
 \cS\os \localvar{\widetilde{T\ x}}{s}\cs_z^G   =   \exists \widetilde{x} (\bigotimes\limits_{i\in 1.. |\widetilde{x}|}\pr{x_i,\nilp,\nonec,\nogroup}  ;\cS\os s \cs_z^G \parallel \texttt{GC})  
}

\noindent
where $
\scriptsize{
\texttt{GC} \defsymbol \pe{z} \to \prod_{i\in 1..|\widetilde{x}| }\texttt{drop}(x_i)
}
$. 
Observe that the freshly created   variables point  to $\nilp$ with no permissions. Once $s$ ends its execution, the local variables are destroyed (definition $
\texttt{GC}$). We note also   that, in this case, we do not add the constraint $\syncp{z}$ nor $\pe{z}$. The reason is that the creation of the local variable can be considered as 
``instantaneous'' and then, the process $\cS\os s \cs_z^G$ will be in charge of marking the termination of the statement.
Note that we ignore the type $T$ since 
our model and analyses are concerned only with the flow of access permissions \red{and we assume that the source program is well typed.} \\

\noindent {\bf Block of statements. } In the block  $\{s_1 \cdots \ s_i \ s_j \cdots s_n\}$,
the process modeling $s_j$ runs in parallel with the other processes once $\cS\os s_{i} \cs^G_{z_i} $ adds the constraint $\syncp{z_{i}}$ to the store. 
Hence, what we observe is that the execution of   $s_j$
is delayed until the encoding of $s_{i}$  has successfully consumed the required permissions. After that, even if $s_{i}$ has not terminated, the encoding of $s_j$ can proceed.
Once $\syncp{z_n}$ can be deduced, 
constraint $\syncp{z}$ is added to the store to synchronize with the
rest of the program.  Moreover, the constraint $\pe{z}$ is added only when all the statements $s_1,...,s_n$ have finished their execution:
 \defrule{\rm R_{COMP}}{
 \cS\os \{s_1 \ ...\ s_i\ ...\ s_n\}\cs_z^G = \wrapp{P,G,z}
}
where $P$ is defined as:
\[
\scriptsize{
\begin{array}{lll}
 P & \defsymbol & \red{\callingp{z}\to \runningp{z}} ; \exists z_1,...z_n (\cS\os s_1\cs_{z_1}^G \parallel   \syncp{z_1}\to \cS\os s_2\cs_{z_2}^G \parallel ...  \parallel   \syncp{z_{n-1}}\to \cS \os s_n\cs_{z_n}^G \parallel  \\
 & &\qquad\qquad\qquad\qquad\qquad\   \syncp{z_n}\to \syncp{z} \parallel    (\red{\runningp{z} \otimes} \bigotimes\limits_{i\in1..n} \pe{z_i}) \to !\pe{z} ) \\
\end{array}
}
\]


\noindent{\bf Groups of permissions}. 
In order to define DGs, we add a constraint specifying that each of those groups has an atomic  DGAP. Recall that the constant $\nostatement$ indicates that the atomic permission is not attached to any particular statement in the program: 
\defrule{R_{NEWG}}{\cS\os \newgroup{g_1,...,g_n}  \cs_{z}^G = \bigotimes\limits_{i\in 1..n} \dgp{g_i,\atomicc,\nostatement}}

Similar to the creation of local variables, we do not mark termination of this statement since it can be considered as ``instantaneous''.   \\

\noindent{\bf Split}.  Let $G' = \{g_1,...,g_m\}$. We define the rule 
for $\texttt{split}$ as follows. 
\defrule{\rm R_{SPLIT}}{\cS \os \splitp{G'}{s_1 \cdots s_n} \cs_z^G =  \wrapp{P,G \setminus G',z}} 
where $P$  and definitions $\texttt{gainP}$, $\texttt{addP}$, $\texttt{exec}$ and $\texttt{restoreP}$ are in Figure \ref{fig:def-split}.
\begin{figure}
\[
\scriptsize{
\begin{array}{rll}
P &\defsymbol& \exists z_1,...,z_n,z'(\texttt{gainP}; \red{\callingp{z}\to} \texttt{addP}; \texttt{exec}; (\syncp{z'} \to\texttt{restoreP});  \red{\runningp{z} \to  !\pe{z} } )\\
\texttt{gainP} & \defsymbol &  \dgp{g_1,\concc,z}\to \envp{g_1,\concc,z} +  \dgp{g_1,\atomicc,\nostatement}\to \envp{g_1,\atomicc,\nostatement} \parallel ... \parallel \\
&&    \dgp{g_m,\concc,z}\to \envp{g_m,\concc,z} + \dgp{g_m,\atomicc,\nostatement}\to \envp{g_m,\atomicc,\nostatement}  \\
\texttt{addP} & \defsymbol & \red{\runningp{z} \otimes} \bigotimes\limits_{i\in 1..n}\bigotimes\limits_{j\in 1..m} \dgp{g_j,\concc,z_i}\\
\texttt{exec} & \defsymbol & \cS \os s_1 \cs_{z_1}^{G'} \parallel \syncp{z_1} \to \cS \os s_2 \cs_{z_2}^{G'} \parallel \cdots \parallel \syncp{z_n} \to \syncp{z'} \\
\texttt{restoreP} & \defsymbol & 
\bigotimes\limits_{i\in 1..n}\bigotimes\limits_{j\in 1..m} \dgp{g_j,\concc,z_i} \to \syncp{z} ; \\
&& 
\envp{g_1,\concc,z} \to \dgp{g_1,\concc,z} + \envp{g_1,\atomicc,\nostatement} \to \dgp{g_1,\atomicc,\nostatement} \parallel \cdots \parallel \\
&& \envp{g_m,\concc,z} \to \dgp{g_m,\concc,z} + \envp{g_m,\atomicc,\nostatement} \to \dgp{g_m,\atomicc,\nostatement} 
\end{array}
}
\]
\caption{Auxiliary definitions for Rule {$\rm R_{SPLIT}$}\label{fig:def-split}}
\end{figure}
Before explaining those definitions, consider the following code:
\begin{lstlisting}
split <g1,g2>{
  s1
  split <g2,g3>{
   s2} }
\end{lstlisting}

\noindent
and assume we are encoding the  $\texttt{split}$ statement in line 3. 
Then, we consider the process $\cS \os \splitp{G'}{s_2} \cs_z^G$ where  $G'=\{g_2,g_3\}$. The set $G=\{g_1,g_2\}$ corresponds to the  concurrent DGAPs assigned by the  external $\texttt{split}$ statement in line 1.  
The process $\texttt{gainP}$  consumes either atomic or concurrent permissions for each DG $g_i \in G'$.  Since such permissions  must be restored once the split command has been executed, we distinguish the case when the consumed permission is concurrent ($\concc$) or atomic ($\atomicc$). For that,   we use the auxiliary predicate symbol (constraint) $\envp{\cdot}$ that keeps information of the DGAP consumed. We note that the DGAP  $g_2 \in G \cap G'$ is consumed and then split again to be assigned to the statement $s_2$. 

Now consider the DG $g_1 \in G\setminus G'$.  Since  $g_1 \not\in G'$,  the DGAP on this group  must be consumed ant it must not be split to be assigned to $s_2$. Hence, the  consumption of any $g \in G \setminus G'$ is  handled by the $\wrapp{\cdot}$ process as in the encoding of other statements.

Once we have consumed the appropriate DGAPs,   we add, for each statement in the block, a concurrent  DGAP  for each of the data groups in $G'$ (definition $\texttt{addP}$). 

The process $\texttt{exec}$ is similar to  Rule $\rm R_{COMP}$ but it uses as parameter $G'$. In our example, this means that concurrent DGAPs on $g_2$ and $g_3$ (and not on $g_1$) are assigned to $s_2$. 
As we already saw in the definition of $\rm R_{COMP}$, the constraint 
 $\syncp{z'}$ is added to the store
 once all the statements in the block were able to consume the required APs. At this point, we wait for all the instructions to reestablish their assigned DGAPs  (definition $\texttt{restoreP}$). Recall that this happens  only when the statements terminate  \red{(see definition \texttt{wrap})}.
 
 Finally, with the help of the constraints 
 $\envp{\cdot}$, we  restore the DGAPs to the environment and we add the constraint $!\pe{z}$ to mark the ending of the block.
  \\

\noindent{\bf Method calls} and {\bf Object instantiation}.  In our encoding  we shall write methods and constructors using functional
notation rather than object-oriented notation. For instance,
$x.m(\widetilde{y})$ is written as $c\_m(x,\widetilde{y})$
when $x$ is an object of type $c$. Similarly, 
the expression $c\_c(x,\widetilde{y})$ corresponds to 
$x:= \texttt{new}\ c(\widetilde{y})$.
 As we shall see, for each method  $m(\widetilde{y})$ of the  class $c$, we shall generate a process definition $c\_m(x,\widetilde{y},z)\defsymboldelta P$. 
The extra argument $z$   is used to later add the constraint $\syncp{z}$ to synchronize with the rest of the program. If $x$ is of  type $c$,  the rule is defined as follows:
\defrule{R_{CALL}}{ \cS\os x.m(\widetilde{y}) \cs_z^G  =  \wrapp{c\_m(x,y_1,..,y_n,z),G,z}  }

The case of the call $x.a.m(\widetilde{y})$ can be 
obtained by using the constraint $\fieldp(\cdot)$ as we did  in Rule  
$\rm R_{ALIAS_F}$ for assignments between fields.

The model of an object initialization is defined similarly but we add also as a parameter the instances of the data groups:

\defrule{ R_{NEW}}{\cS\os x:= \texttt{new} \ {c\langle g_1,...,g_n\rangle(\widetilde{y})}\cs_z^G =   \wrapp{{c\_c}(x,\widetilde{y},z,g_1,...,g_n),G,z} 
}
 
\subsection{Modeling  Class Definitions.} \label{sec:defs}
In this section we  describe  function $\cD\os \cdot \cs$
interpreting method and constructors definitions as \lcc\ process definitions. \\

\noindent{\bf Method Definitions.}
Let $
\red{m (\widetilde{c_y\langle\widetilde g_y \rangle\ y})  \ \ p(\thisp),\widetilde{p(y)}\To p'(\thisp),\widetilde{p'(y)}\  \{ s \} }
$ 
be a  method in class $c\langle\widetilde g_x \rangle$. We define
\defrule{\rm R_{MDEF}}{
 \cD \os c.m\cs =   {c\_m}(x,\widetilde{y},z)\defsymboldelta \exists \widetilde{y'},x' ( Consume   ; \syncp{z} ; \red{\callingp{z}\to \runningp{z};Body})
}
where $n= |\widetilde{y}|=|\widetilde{y'}|$, 

\resizebox{\textwidth}{!}{
\red{$
\begin{array}{lll}
Consume  & \defsymbol &  \prod\limits_{i\in 1..n}\texttt{consume}(y_i,y_i',p_i, c)  \parallel  \texttt{consume}(x,x',p, c)\\
Body &  \defsymbol  & \exists z' (\cS \os\widehat{s} \cs_{z'} \parallel  (\syncp{z'}\otimes \pe{z'}) \to \\
&& \ \ \ \ \tabs\tabs\tabs\tabs (\restenv(x, p, x', p', c) \parallel 
\prod\limits_{i\in 1..n} \!\!\!\restenv(y_i, p_i,y_i', p_i', c))) \ ;\runningp{z}\to !\pe{z}  
\end{array}
$}}
and the auxiliary process definitions $\consumep(\cdot)$ and $\restenv(\cdot)$ 
are in Figure \ref{fig:con-res}. 

In the process definition $c\_m(x, \widetilde{y}, z)$, the first   parameter $x$ represents 
the object caller $\texttt{this}$ and the last parameter $z$ is used for synchronization. This definition  first 
declares the local variables $\widetilde{y'}$ and $x'$
to replace the formal parameters ($\widetilde{y}$) and the caller ($x$) by the actual parameters. Next, it 
consumes the required permissions from  $\widetilde{y}$ and  $x$,
and assigns them to the previously mentioned local variables. Finally, the constraint $\syncp{z}$ is added and the encoding of the method's body is executed. \red{In the following we explain the definitions $Consume$ and \emph{Body}. }

\begin{figure}
$\consumep(x,x',p, cname) \defsymbol$

\resizebox{\textwidth}{!}{
$
\begin{array}{lll}
&  &  \left\{ 
  \begin{array}{l}
\forall o (\pr{x,o,p,\nogroup} \ot  \countp{o,n} \to  \pr{x,o,p,\nogroup} \ot  \countp{o,s(n)}
\ot \pr{x',o,p,\nogroup})
\mbox{ if } p =\immc  \\\\
\forall g,o (\gparam{cname\_g,o,g} \to
\pr{x,o,p,g} \ot  \countp{o,n} \to  \pr{x,o,p,g} \ot  \countp{o,s(n)}) \ot \pr{x',o,p,g}  \mbox{ if } p = \shrc:g
 \\\\
\forall o(\pr{x,o,p,\nogroup}  \to  \pr{x',o,p,\nogroup})   \mbox{ if } p \in \{\unqc,\nonec\}
 \end{array}
\right.
\end{array}
$
}

$ \restenv(x,p, x',p', cname)\defsymbol$
$
\scriptsize{
\left\{
\begin{array}{lll}
  & \forall o',n(
  \pr{x',o',p',\nogroup}  \ot \countp{o',s(n)}  \to   \countp{o',n}) \mbox{ \ \ \ if } p,p' = \immc \\\\
   &   \forall o',n,g(\gparam{cname\_g,o',g} \to   \pr{x',o',p',g}  \ot \countp{o',s(n)}  \to  \countp{o',n}) \mbox{ \ \ \ if } p,p' = \shrc:g \\\\
   &\!\!\! \forall o'(
  \pr{x',o',p',\nogroup}\otimes \countp{o',s(\zero)}    \to   \pr{x,o',p',\nogroup} \otimes \countp{o',s(\zero)} \mbox{ \ \ \ if } p, p' = \unqc\\\\
  &\!\!\! \forall o'(
  \pr{x',o',p',\nogroup}  \to   \pr{x,o',p',\nogroup}  \mbox{ \ \ \ if } p =\nonec \\\\
   &\!\!\! \forall o,n,o'(
 (\pr{x,o,p,\nogroup}\ot \countp{o,s(n)}  \to \\
  & \tabs  \countp{o,n});  \pr{x',o',p',\nogroup} \to \pr{x,o',p',\nogroup}) \mbox{ if } p=\immc, p'\in \{\unqc,\nonec\}\\\\
    &\!\!\! \forall o,n,o',g(
 (\gparam{cname\_g,o,g}\ot\pr{x,o,p,g}\ot \countp{o,s(n)}  \to \\
  & \tabs  \countp{o,n}); \pr{x',o',p',\nogroup} \to \pr{x,o',p',\nogroup})  \mbox{ if } p=\shrc, p'\in \{\unqc,\nonec\}
 \end{array}
 \right.
 }
$
\caption{Auxiliary definitions for constructor and method declarations. \label{fig:con-res}}
\end{figure}

The definition of $\texttt{consume}(x,x',p, c)$ in Figure \ref{fig:con-res} can be read as  ``\emph{consume the permission $p$ on the variable $x$ and assign it to the variable $x'$}. 
If the required permission is share or immutable, the permission is split and restored  to allow concurrent executions in the environment that called the method. We recall that in  $p=\shrc:g$, $g$ must be a data group parameter in the class $c$. This explains the last parameter in 
$\texttt{consume}(\cdot)$. We then use the predicate 
$\bang\gparam{c\_g,o,g}$, added by   the encoding of constructors, as we shall see,   to establish the link between the DG parameter and the current DG. Finally,  unique and none permissions are consumed and transferred to the local variables. 

Now we focus on the definition $Body$ where  $\widehat{s}$ denotes $s$ after replacing $y_i$ by 
$y_i'$ and $x$ by $x'$.  Once    $\widehat{s}$ finishes (i.e., it adds $\pe{z'}$ to the store), the references and permissions of the local variables created to handle the parameters are consumed and restored to the environment according to $ \restenv(x,p, x',p', c)$ in Figure \ref{fig:con-res} (consume the permission $p$ on   $x$ and transforms it into a permission $p'$ to the variable $x'$).  Let us give some intuition about the cases considered in  this definition.  Recall that $\consumep$ \emph{replicates} the
$\shrc$ and $\immc$ permissions for the variables internal to the
method. Therefore, we only need to consume those permissions and decrease
the number of references pointing to object $o'$. When the input permissions are $\unqc$ or
$\nonec$,
$\consumep$ \emph{transfers} those
permissions to the local variables and \emph{consumes} the external
references. Then, $\restenv$ needs to restore the external reference
and consume the local one (the number of references pointing to $o'$
remains the same). When the method changes the input permission from
share or immutable into a unique or none, we need to \emph{consume}
first the external reference. Afterwards, we \emph{transfer} the internal
permission and reference to the external variable.


\noindent{\bf Constructor definitions. }
Let  $
c (\widetilde{c_y\lr{\widetilde{g_x}}\ y})  \ \ \pn{\thisp},\widetilde{p(y)}\To p'(\thisp),\widetilde{p'(y)}\  \{ s \} 
$  be a constructor of a \red{parameterized class $c \langle  pg_1,...,pg_k\rangle$}. We define 
\defrule{\rm R_{CDEF}}{
\red{
\!\!\!\begin{array}{lll}
  \cD\os C_D\cs  =  c(x,\widetilde{y},z,g_1,...,g_k) &  \defsymboldelta& \exists \widetilde{y'}, x', o_{new} (  \texttt{gparam-init} ; \consumep'   ; \\
  && \ \ \ \exists \widetilde{u}( \texttt{fields-init}  \ ; \syncp{z} ;\callingp{z}\to\runningp{z};\\
&& \ \ \ \ \   \exists z' (\cS\os  \widehat{s}\cs _{z'} \parallel (\syncp{z'} \ot \pe{z'}) \to\\
&& \ \ \ \ \ \ \ \   (\restenv(x, p,x',p' ,c ) \parallel  \prod\limits_{i\in 1..m}  \restenv(y_i, p_i, y_i', p_i', c)) )) ;\\
&& \ \  \ \ \ \ \ \ \ \runningp{z}\to !\pe{z})\\\\
  \end{array}
}}
 \red{where $n= |\widetilde{y}|=|\widetilde{y'}|$ and }

\resizebox{\textwidth}{!}{
\red{$
\begin{array}{lll}
 \consumep'&  \defsymbol  & \prod\limits_{i\in 1..m}\texttt{consume}(y_i,y_i',p_i, c)  \parallel \\
 &&  \pr{x,\nilp,\nonep,\nogroup}  \to  \pr{x',o_{new},\unqc,\nogroup} \otimes \countp{o_{new},s(\zeroc)}  \\
 \texttt{gparam-init} & \defsymbol & 
\bigotimes\limits_{i\in 1..k}\bang\gparam{c\_{pg_i},o_{new},g_i}\\
\texttt{fields-init} &\defsymbol&
\bangp\fieldp(u_1,o_{new},c \_a_1) \ot\pr{u1,\nilp,\nonec,\nogroup} \ot...\ot \\&&\bangp\fieldp(u_k,o_{new},c\_a_k) \ot   \pr{u_k,\nilp,\nonec,\nogroup}
\end{array}
$}}


The mechanisms for parameter passing, executing the body $\widehat{s}$ and  restoring permissions are the same as in method definitions. The definition $\texttt{consume'}$ is similar to  $\texttt{consume}$ in method definitions 
but, instead of using $consume(x,x',p,c)$, we consume the constraint $\pr{x,\nilp,\nonep,\nostatement}$,
 i.e., $x$ in the statement $x:= \texttt{new} \ c \lr{\widetilde{g}}(\widetilde{y})$ is
restricted to be a null reference. Moreover, the internal variable $x'$ points to the newly created object  $o_{new}$ with permission unique.

The definition $ \texttt{gparam-init}$ allows us to  establish the link between the new object $o_{new}$ and the group parameters. In the constraint $\gparam{c\_{pg_i},o_{new},g_i}$,  
 the constant symbol  $c\_pg_i$ corresponds to the name defined for the DG parameter $pg_i$ of the class $c \langle  pg_1,...,pg_k\rangle$ and  $g_i$ is the current DG passed as parameter.

The initialization of fields is controlled by  the definition $\texttt{fields-init}$. 
The added constraint $\fieldp(u_i,o_{new},c \_a_i)$ 
establishes the link between 
the field $o_{new}.a_i$ and  the null reference $u_i$.


Let  us present a couple of examples to 
show  the proposed model in action. 

\begin{figure}
\begin{lstlisting}[firstnumber=8]
main(){
   let collection c, stats s in
     c := new collection()
     s := new stats()
     c.compStats(s)
     c.compStats(s)
     c.removeDuplicates()
 end}
\end{lstlisting}
\caption{Main program for Example \ref{ex:1}. Class definitions are in Figure \ref{fig:ae-code} \label{fig-ex-AP}}
\end{figure}

\begin{example}[Access Permission Flow]\label{ex:1}
Assume the class definitions $stats$ and $collection$ in Figure \ref{fig:ae-code} and the main body in Figure \ref{fig-ex-AP}. 
The \lcc\ agent modeling the statement in line $10$ calls  $collection\_collection(c,z_{10})$, which triggers the execution of the body of the constructor   (see Rules $\rm R_{CDEF}$  and $\rm R_{CALL}$). Variable $z_{10}$
is the local variable used to synchronize with the rest of the program
(see Rule $\rm R_{COMP}$).  Once the agents modeling the statements in
lines $10$ and $11$ are executed, the following store is observed:
 \[\!\!\!
 \begin{array}{c}
 \exists c,s,o_c,o_s( \pr{c,o_c,\unqc,\nogroup} \!\ot\!\pr{s,o_s,\unqc,\nogroup} \!\ot\!  \countp{o_c,s(\zeroc)}\ot \countp{o_s,s(\zeroc)})
 \end{array}
 \]
 Hence, $c$ (resp. $s$) points to $o_c$ (resp. $o_s$) with a unique
 permission. In \lstinline{c.compStats()},   $c$ requires 
 an immutable permission to $o_c$. The axiom $downgrade_2$ is used to
 entail the guard of $\texttt{consume}$ in the definition of the
 method.  Let $c'$ be the representation of
 $c$ inside the method (see Rule $\rm R_{MDEF}$).
 We notice that when the body of the method is being executed, both
 $c$ and $c'$ have an immutable permission to $o_c$, i.e., the store contains the tokens
 \[
 \pr{c,o_c,\immc,\nogroup} \ot
  \pr{c',o_c,\immc,\nogroup} \ot
  \countp{o_c,s(s(\zeroc))}
 \]
Before executing
 the body of method $compStats$  constraint $\syncp{z_{12}}$ is added,
 so as to allow possible concurrent executions in the main body (see Rule $\rm R_{COMP}$). 
 Hence, the 
 agent modeling the statement in line $13$ can be 
 executed and we have a store with three references with immutable permission to object $o_c$, namely, $c$, $c'$ as before, and  $c''$, the representation of $c$ inside the method $print$.   Now, once  constraint $\syncp{z_{13}}$ is added by the definition of $print$, the process representing the statement in line 14 can be executed. 
 However, this call requires $c$ to have a unique permission
 to $o_c$ which is not possible since the axiom $upgrade_2$ requires
 that $c$ is the sole reference to $o_c$. Hence, the guard
 $\texttt{consume}$ for this call is delayed (synchronized) until the
 permissions on $c'$ and $c''$ are consumed and restored to the environment (see
  $\restenv$ in Rule $\rm R_{MDEF}$).  We then observe that 
  statements in lines 12 and 13 can be executed concurrently but the statement in line 14 is delayed until the termination of the previous ones. 
  \end{example}

  
  \begin{example}[Data Group Permissions Flow]\label{ex:dgp-flow}
  Now consider the program in Figure \ref{fig:s-o-DG}. Figure  \ref{tab:exemple-DG}  shows the stores  generated by the model of this program. We omit some tokens for the sake of readability.  
  \end{example}
  
  \begin{figure}
  \resizebox{.7\textwidth}{!}{
  \begin{tabular}{| l | p{3.9cm} | p{7.8cm}  |}
  \cline{1-3}
  {\bf Line} & {\bf Store} & {\bf Observations} \\
  \cline{1-3}
 7 & $\dgp{g,\atomicc,\nostatement}$& See Rule ${\rm R_{NEWG}}$.  \\
  \cline{1-3}
 8 & $\dgp{g,\atomicc,\nostatement}  \ot \pr{s,\nilp,\nonec,\nogroup} \ \  \ot \pr{o1,\nilp,\nonec,\nogroup}\ot \pr{o1,\nilp,\nonec,\nogroup} $
 &  $s$, $o1$ and $o2$ are null references (see Rule ${\rm R_{LOC}}$).  \\
  \cline{1-3}
 10& $\dgp{g,\concc,z_{10}}\ot \dgp{g,\concc,z_{11}}\ot  \pr{s,\nilp,\nonec,\nogroup} \ot  \cdots  $ & The $\atomicc$ DGAP on $g$ is consumed and split into  $\concc$ permission for statements in lines 11-12 (see Rule $\rm R_{SPLIT}$) \\
   \cline{1-3}
  Before 13(1) & $\dgp{g,\concc,z_{10}}\ot  \cdots \ot \pr{s,o_s,\shrc,g} \ot \pr{o1,oo_1,\unqc,\nogroup}\ot \pr{o2,oo_2,\unqc,\nogroup} $ & Variables $s$ and $obs$ are instantiated. The atomic DGAP has not been restored yet and then, statement in line 13 has to wait. \\
   \cline{1-3}
     Before 13(2) & $\dgp{g,\atomicc,\nostatement}\ot  \pr{s,o_s,\shrc,g} \ot \cdots  $ & Concurrent DGAPs are consumed and the atomic  permission on $g$ is restored (see Rule $\rm R_{SPLIT}$).  \\
   \cline{1-3}
   Before 16 & $\dgp{g,\concc,z_{14}}\ot \dgp{g,\concc,z_{15}} \ot \pr{s,o_s,\shrc,g} \ot \pr{s',o_s,\shrc,g} \ot 
   \pr{s'',o_s,\shrc,g} \ot \pr{o1,oo_1,\unqc,\nogroup}\ot \pr{o2,oo_2,\unqc,\nogroup} $ & There are 3 references to $o_s$: $s$,   $s'$ and  $s''$. The last two   correspond to the internal representation of $s$ in the calls to method  $update$ (see  Rule Rule $\rm R_{MDEF}$). Then, such methods can be executed concurrently. We also see that the $\atomicc$ DGAP was split into $\concc$ DGAP for statements 14 and 15. \\
   \cline{1-3}
   16 &$\dgp{g,\atomicc,\nostatement}\ot \pr{s,o_s,\shrc,g}  \ot \pr{o1,oo_1,\unqc,\nogroup}\ot \pr{o2,oo_2,\unqc,\nogroup} $  
   &  In the end, $s$ is the sole reference to $o_s$ (see $\texttt{r\_env}$ in Rule ${\rm R_{MDEF}}$)   and the atomic DGAP on $g$ is reestablished. \\
   \cline{1-3}
  \end{tabular}
  }
  \caption{Constraints added by the processes in Example \ref{ex:dgp-flow} \label{tab:exemple-DG} (AP code in Figure \ref{fig:s-o-DG})}
  \end{figure}
 


\begin{example}[Deadlocks]\label{blind_call}
Let us consider the following implementation   for the method $compStats$ in  the class $collection$  (see Figure \ref{fig:ae-code})

\begin{lstlisting}
compStats(s)  imm(this), unq(s) => imm(this), unq(s) { 
   ...  
   c.sort() 
   ... }
\end{lstlisting}

Consider the call  \lstinline{c.compStats(s)}  and suppose that, in the \lcc\ model,   variable $c$ points to the object $o_c$. 
When the $compStats$ method  is invoked, the immutable permission is divided between the external reference $c$ and the internal reference $c'$. For this reason, inside the method, reference $c'$ cannot acquire a unique permission for the  invocation of method $sort$ which then blocks. 
Our analysis will thus inform that there is a deadlock, unless, e.g., 
the program includes the statement  $c\lr{g}:=\nilp$ to discard  the permission of 
 $c$  to $o_c$. 

Consider now the following definition of the same method:
\begin{lstlisting}
compStats(s)  unq(this), unq(s) => unq(this), unq(s) {
  ... 
  c.sort() 
 ...}
\end{lstlisting}

When  $compStats$   is invoked, the unique permission is transferred from reference $c$ to (the internal) reference $c'$. The invocation of method $sort$ has thus the right permissions to be executed  and it does not block. 
\end{example}
%
%

\subsection{The  Model as a Runnable Specification}\label{sec:tool}
Models based on the \ccp\ paradigm can be regarded as runnable
specifications, and so we can observe how permissions evolve
during program execution by running the underlying \lcc\
model. We implemented an interpreter of \lcc\ in Java and  used Antlr (\url{http://www.antlr.org}) to generate a parser from  AP programs  into \lcc\ processes following our encoding. The resulting  \lcc\ process is then executed  and    a program trace is output. The interpreter and
the parser have been integrated into Alcove  (Access Permission Linear COnstraints VErifier) Animator, a web application freely available
at \url{http://subsell.logic.at/alcove2/}. The URL further
includes all the examples presented in this section. In the following we explain some outputs of the tool. 



\begin{example}[Trace of Access Permissions]\label{ex:parallel} The program in Figure \ref{fig:ae-code} generates the trace depicted in Figure \ref{fig:trace1}.
For verification purposes, 
the implementation extends the predicates $\callingp{\cdot}$, $\runningp{\cdot}$ and $\donep{\cdot}$ to include also the variable that called the method, the name of the method and the number of line of the call.  Note for instance that the call to $print$ (line 9 in Fig. \ref{fig:trace1}) was marked while the method $sort$ was running (line 7). Nevertheless, the execution of $print$  (line 11) must wait until $sort$ terminates (line 10).  In this trace, the constructor $stats$ (line 5) runs in parallel with $sort$ (line 7).  Finally, the execution of $removeDuplicates$ (line 17) is delayed until the methods $print$ (line 13) and $compStats$ (line 16) terminate. 
Lines 20 and 21  show that both $c$ and $s$  end with a unique  permission to objects \verb|o_4774|
and \verb|o_79106|, respectively (the numbers that follow the variable names
  are generated each time a local variable is created to avoid clash
  of names). 
\end{example}

\begin{figure}
\begin{lstlisting}
act(C_628,collection_collection,line 10 (Z_PAR_814))
run(C_628,collection_collection,line 10 (Z_PAR_814))
act(S_729,stats_stats,line 11 (Z_PAR_915))
end(C_628,collection_collection,line 10 (Z_PAR_814))
run(S_729,stats_stats,line 11 (Z_PAR_915))
act(C_628,collection_sort,line_12 (Z_PAR_1016))
run(C_628,collection_sort,line_12 (Z_PAR_1016))
end(S_729,stats_stats,line 11 (Z_PAR_915))
act(C_628,collection_print,line_13 (Z_PAR_1117))
end(C_628,collection_sort,line_12 (Z_PAR_1016))
run(C_628,collection_print,line_13 (Z_PAR_1117))
act(C_628,collection_compStats,line_14 (Z_PAR_1218))
end(C_628,collection_print,line_13 (Z_PAR_1117))
run(C_628,collection_compStats,line_14 (Z_PAR_1218))
act(C_628,collection_removeDuplicates,line_15 (Z_PAR_1319))
end(C_628,collection_compStats,line_14 (Z_PAR_1218))
run(C_628,collection_removeDuplicates,line_15 (Z_PAR_1319))
end(C_628,collection_removeDuplicates,line_15 (Z_PAR_1319))
  
[ref(C_628,O_4774,unq,ng), ct(O_4774,1)]
[ref(S_729,O_79106,unq,ng), ct(O_79106,1)]
ok()
567 processes Created 
\end{lstlisting}

\caption{Trace generated by the program in Figure \ref{fig:ae-code}
\label{fig:trace1} (Example \ref{ex:parallel})}
\end{figure}

\begin{example}[Deadlock Detection]\label{ex:deadlock} 
  Let  us assume now the class definitions in Figure  \ref{fig:ae-code} and the  
  following \verb|main|:
  \begin{lstlisting}[firstnumber=8]
main(){ 
 group<g>
 let collection c, stats s, stats svar in
    c := new collection() 
    s := new stats()
    svar<g> := s
    c.compStats(s) 
    c.compStats(svar)
end}
\end{lstlisting}

   The assignment in line 13 aliases  $svar$ and $s$ so they share the same permission
  afterwards.
  Therefore, $s$ cannot recover the unique permission to  execute the statement  in line 14, thus leading to a permission deadlock. This bug is detected by \thetool\ as depicted in Figure \ref{fig:deadlock} (line 13). Observe in the trace that  $compstats$ is called  (line 7 in the trace)  but  not executed. Furthermore, both $s$ and $svar$ have a share permission on the same pointed object (lines 17 and 18). Moreover, both   $c$ (\verb|c_644|) and its internal representation inside   $compStats$ (\verb|inner_136172|) have an immutable permission on object $\verb|o_6491|$ (lines 16 and 19). 
  Lines 8-11 show   the suspended \lcc\ processes in the end of the computation that were killed by the scheduler. Particularly, line 10 shows that there is an ask agent trying to consume a unique permission on object \verb|O_142| pointed by \verb|S_745|. 
\end{example}


\begin{figure}  
\begin{lstlisting}
act(C_644,collection_collection,line 10 (Z_PAR_928))
run(C_644,collection_collection,line 10 (Z_PAR_928))
act(S_745,stats_stats,line 11 (Z_PAR_1029))
run(S_745,stats_stats,line 11 (Z_PAR_1029))
end(C_644,collection_collection,line 10 (Z_PAR_928))
end(S_745,stats_stats,line 11 (Z_PAR_1029))
act(C_644,collection_compStats,line_13 (Z_PAR_1231))
[Killed] ask endc(line_13 (Z_PAR_1231)) then ...
[Killed] ask sync(line_14 (Z_PAR_1332)) then ...
[Killed] ask ref(S_745,O_142,unq,ng) then ... + ask 
...
404 processes Created
[FAIL] Token ok not found. End of the program not reached.

VARIABLES
C_644 -> O_6491. imm:ng
S_745 -> O_96123. shr:GRP_461
SVAR_846 -> O_96123. shr:GRP_461
INNER_136172 -> O_6491. imm:ng
\end{lstlisting}
\caption{Trace generated by the program in Example \ref{ex:deadlock}
\label{fig:deadlock}}
\end{figure}

\subsection{Adequacy of the Encoding}
\red{In this section  we present some invariant properties of the encoding and prove it correct. There are three key arguments in our proofs:}

\begin{observation}[Ask agents]\label{obs:ask}
\red{ (1) the ask agents controlling both the APs  (Proposition \ref{obs-inv}) and the state of statements (Proposition \ref{obs-state}) are of the form $c\to P$ where $P$ is a tell agent (and not, e.g., a parallel composition). Hence, in one single transition, the encoding  consumes and produces the needed tokens to maintain the invariants (ruling out intermediate states where the property might not hold). Moreover, (2) such ask agents are preceded  by the sequential composition operator ``;''. This means that, before consuming the needed constraints, some action must have been finished.
In particular, (3)  the ask agent $\callingp{z} \to \runningp{z} $
is executed only when the needed permissions are consumed  and  the ask agent $ \runningp{z}\to \donep{z} $ is executed only after restoring the consumed permissions (Rules ${\rm R_{ALIAS}}$, ${\rm R_ {CDEF}}$ and ${\rm R_ {MDEF}}$).
}
 \end{observation}

The following invariants show that the \lcc\ model correctly keeps track of the variables and their corresponding pointed objects. 
\begin{proposition}[Invariants on References]\label{obs-inv}
Let $S$ be an AP program and $\mathcal{D}.P$ its corresponding translation into \lcc. Assume that $(\emptyset ; P ; \one) \redi^*  (X ; \Gamma ; c)$. The following holds: 
\begin{enumerate}
 \item If $c \entails \pr{x, o , \unqc, \nogroup}$ then $c \entails \countp{o,s(\zero)}$. 
 \item If $c \entails \pr{x, o ,p, g}$ and $p \in \{\shrc, \immc\}$ then, there exists $n>0$ s.t.  $c \entails \countp{o,n}$. 
 \item If $c \entails \pr{x, o ,p, g}$  and  $c  \entails \pr{x, o' ,p', g'}$ then  $o'=o$, $p' = p$  and $g' = g$.
 \item  If $c \entails \pr{x, \nilp ,p, g}$ then $p=\nonep$ and $g=\nogroup$. 
 \item \red{($\bf{counting}$) 
if $c \entails \countp{o, n}$ then:}
\begin{itemize}
\item[(a)] for all $m\leq n$,  $c \entails \exists x_1,p_1,g_1...,x_m,p_m,g_m  
\bigotimes\limits_{i\in 1..m} \pr{x_i,o,p_i,g_i}$; and
\item[(b)] for all $m> n$, $c \not\entails \exists x_1,p_1,g_1...,x_m,p_m,g_m  
\bigotimes\limits_{i\in 1..m} \pr{x_i,o,p_i,g_i}$
\end{itemize}  
\end{enumerate}
\end{proposition}
\begin{proof}
\red{
An inspection of the encoding reveals  that  the rules ${\rm R_{ALIAS}}$ and ${\rm R_{LOC}}$ and the definitions $\texttt{consume}$, $\restenv$ and $\texttt{fields-init}$
are the only ones that consume/produce $\pr{\cdot}$ and $\countp{\cdot}$ constraints. 
For any newly created variable, 
${\rm R_{LOC}}$  and  $\texttt{fields-init}$ add the needed $\pr{\cdot}$
token adhering to item 4. Moreover, the ask agents in the above rules/definitions adhere to the conditions in Observation \ref{obs:ask}. Therefore,   
if the agent $c\to P$  consumes   a constraint of the form $\pr{x,o,p,g}$, the {\bf tell} process $P$  adds  the needed constraints to maintain  correct the counting of references to $o$. 
 }
\end{proof}

The next proposition shows that the encoding correctly captures the state of statements. 
\begin{proposition}[States]\label{obs-state}
\red{Let $State=\{\callingpn, \runningpn, \donepn\}$, $S$ be an AP program and  $\mathcal{D}.P$ its corresponding \lcc\ translation. Consider an arbitrary execution starting at $P$:}
\[
( \emptyset ; P  ; 1 ) \redi 
( X_1 ; \Gamma_1  ; c_1 ) \redi
( X_2 ;  \Gamma_2  ; c_2 ) \redi \cdots \redi
( X_n ;  \Gamma_n  ; c_n )
\]
\red{Let $z\in X_n$,  $st\in State$, $x\in 1\ldots n$ and assume that $c_x \entails st(z)$. 
Then, }
\begin{enumerate}
 \item \red{({\bf no confusion})  for all $st'\in State \setminus \{st\}$, $c_x \notentails st'(z)$.  }
 
 \item\red{ ({\bf state ordering}) there exists $i\in 1..x$ such that}
\begin{enumerate}
\item \red{({\bf init}) for all $k \in [1,  i)$ and $st'\in State$,  $c_k \notentails st'(z)$. }
\item \red{({\bf continuity})  for all $k \in [i,n]$, $c_k \entails st'(z)$ for some $st'\in State$.}
\item \red{ ({\bf act}) if  $c_n \entails \callingp{z}$  then  for all $k\in [i,n]$, $c_k \entails \callingp{z}$. }
 \item  \red{({\bf act until run}) If $c_n \entails \runningp{z}$ then, there exist  
 two non-empty intervals $A= [i...j_r)$ and $R=[j_r,n]$ s.t. 
 for all $k\in A$, $c_k \entails \callingp{k}$ and
 for all $k\in R$, $c_k \entails \runningp{k}$. }
 \item  \red{({\bf run until end}) If $c_n \entails \donep{z}$ then there are 3 non-empty intervals 
$A=[i..j_r)$, $R=[j_r,...,j_e)$,  $E=[j_e,...,j_n]$ s.t. $A$ and $R$ are as above  and   for all $k\in E$, $c_k \entails \donep{z}$. }
 \end{enumerate}
\end{enumerate}
\end{proposition}
\begin{proof}
\red{Note that the token $\callingpn(z)$ is added when the encoding of a statement is activated (\texttt{wrap}). An inspection of the encoding shows that the ask agents controlling the state of statements  adhere to  conditions in Observation \ref{obs:ask}. 
Since each executed statement uses a freshly created variable $z$ (see $\rm R_{COMP}$), we can show that,  for any $z$ and multiset  $\Gamma_x$, 
$\Gamma_x$ can contain at most one of each of such ask agents (using $z$). 
 Hence, for all $st\in State$,  if $s(z)$ is consumed 
 from the store $c_x$, the store $c_{x+1}$ must contain the next state $st'(z)$. 
This guarantees the correct ordering of states. }
\end{proof}

\red{We conclude by showing that the encoding enforces the execution of statements  according to the AP specification. More precisely, 
the activation of a statement $s$ is delayed until  
its (lexical) predecessor has successfully consumed the needed permissions;  the execution of $s$  is delayed until its required permissions are available (and consumed); signalling the termination of $s$ is delayed until all the consumed permissions are restored. }

\begin{theorem}[Adequacy]\label{adequacy1}
Let $S$ be an AP program and  $\mathcal{D}.P$ its corresponding \lcc\ translation. Let $s_i$ and $s_j$ be two  sentences 
that occur in the same block and $s_j$ is lexically after $s_i$. Then, 
\begin{enumerate}
 \item (safety)  $s_i$ and $s_j$ are in conflict iff for any reachable configuration $(X ;  \Gamma ; c)$ 
\red{from $(X ; P ; 1)$},   $c \entails \runningpn(z_{s_j})$ implies $c \entails \donep{z_{s_i}}$.
\item (concurrency)   $s_i$ is not in conflict with $s_j$ iff  there exists a reachable configuration $(X ;  \Gamma ; c)$ \red{from $(X ; P ; 1)$} s.t. $c \entails \runningpn(z_{s_i})$ and  $c \entails \runningpn(z_{s_j})$. 
\end{enumerate}

\end{theorem}
\begin{proof}
\red{
The execution of assignments, the call to methods/constructors
and the beginning of blocks are the statements we have to synchronize in the encoding. Note that rules 
 ${\rm R_{ALIAS}}$, ${\rm R_ {CDEF}}$,  ${\rm R_ {MDEF}}$, ${\rm R_{SPLIT}}$ and ${\rm R_{COMP}}$ adhere to conditions in Observation \ref{obs:ask}. In particular, condition (3) shows that the changes of states are controlled by acquiring / releasing permissions. 
}

\red{($\Rightarrow$) 1. 
Assume that $s_i$ and $s_j$ both require a unique permission on the same object (the other kind of conflicts are similar). 
From rule ${\rm R_{COMP}}$, we know that $s_i$ first consumes its permissions (before enabling $s_j$). From Propositions \ref{obs-inv} and \ref{obs-state}  we can show that $s_j$ cannot 
move to the state $\runningpn$  until $s_i$ moves to state $\donepn$. }

\red{2. If there are no conflicting resources, then both processes may successfully consume the needed permissions from the store. Consider the  following trace: the encoding of $s_i$ consumes the needed permissions, adds $\runningp{z_{s_i}}$ and the $\syncp{z_{s_i}}$ token. Then, the encoding of $s_2$ can start its execution (consuming $\syncp{z_{s_i}}$), consumes the needed permissions and adds   $\runningp{z_{s_j}}$ to the store.  }

\red{$(\Leftarrow)$ For (1), assume that  in any reachable configuration $(X ; \Gamma ; c)$,  $c \entails \runningpn( z_{s_j})$ implies $c \entails \donep{z_{s_i}}$. By 
Proposition \ref{obs-state} we know that 
$c\notentails \runningp{z_{s_i}}$. Since the encoding maintains correct the number of references (in the sense of Proposition \ref{obs-inv}), there is no reachable store able to entail the permissions needed for both $s_i$ and $s_j$. Hence,   there is a conflicting access in $s_i$ and $s_j$.  The case (2) follows from a  similar argument. }
\end{proof}

%
%
%
%
%
 

\section{Logical Meaning of Access Permissions}\label{sec:verif}
\label{sec:verification}
Besides playing the role of executable specifications,   \ccp-based models can be declaratively interpreted  as formulas in logic
\cite{cp-book,DBLP:journals/toplas/BoerGMP97,NPV02,fages01linear,TCS16}. This section provides additional mechanisms and
tools for verifying properties of AP based
programs. More concretely, we take the \lcc\ agents generated from the AP program  and translate them as an  
intuitionistic linear logic (ILL)
formula. Then, a property specified in ILL is verified with the
\thetool\ LL Prover, a  theorem prover   implemented on top of
Teyjus (\url{http://teyjus.cs.umn.edu}), an implementation of 
$\lambda$-Prolog~\cite{DBLP:conf/iclp/NadathurM88,lProlog}. 

\red{Our analyses are based on reachability properties, i.e., we  verify  the existence of reachable \lcc\ configurations satisfying some conditions.
It turns out that  this is enough for verifying interesting properties of AP programs. For instance,  we can check whether a program is dead-lock free or whether two statements can be executed concurrently. 
 }

\subsection{Agents as Formulas}\label{ag-as-for}
The logical interpretation of \lcc\ agents as formulas
in intuitionistic linear logic ILL 
is defined with the aid of a function $\cL\os \cdot \cs$ defined in Figure \ref{fig:lcc-to-ill} 
\cite{fages01linear}.
As expected, parallel composition is identified with multiplicative conjunction and  ask processes correspond to linear implications. 
Moreover, process definitions are (universally quantified) implications to allow the unfolding of its body. 

%


\begin{figure}
$$
\begin{array}{lll l lll}
\cL\os c \cs& = & c &\quad & 
\cL\os P\parallel Q \cs& = & \cL \os P \cs \ot \cL \os Q \cs\\
\cL\os\sum\limits_{i\in I} \forall \vx_i(c_i \to P_i)\cs& = &
\with_{i\in I} (\forall \vx_i(c_i \lolli \cL\os P \cs)) & \quad &
\cL\os \exists x(P) \cs& = & \exists x.( \cL\os P \cs)\\
\cL \os p(\widetilde{x}) \defsymboldelta P \cs & = &
\forall \widetilde{x}.p(\widetilde{x})\lolli  \cL \os P \cs & \quad &
\cL\os p(\vx)\cs &=& p(\vx)
\end{array}
$$%
\caption{Interpretation of \lcc\ processes as ILL formulas.  \label{fig:lcc-to-ill}}
\end{figure}

%

In what follows, we will show how to use
logic in order to have a better
control of the operational flow and, therefore,
be able to verify properties
of AP programs. 

The first step
consists of interpreting the \lcc\ model  in Section \ref{sec:encoding} as 
ILL formulas via  $\cL \os \cdot \cs$. 
We shall call {\em definition clauses}  to the encoding of   process definitions of the form  $p(\widetilde{x}) \defsymboldelta P$  (i.e., assignment and constructor and method definitions in our encoding)
 and we shall include them 
in a theory
$\Delta$,
together 
with the  axioms of
 upgrade and downgrade in Figure \ref{fig:signature}. 
%
The next example illustrates this translation. For the sake of readability, we shall omit  empty synchronizations such as  $\syncp{z} \otimes (\syncp{z} \to \one)$. 

\begin{figure}
\resizebox{.75\textwidth}{!}{
$
\scriptscriptstyle
\red{
\begin{array}{lll}
\multicolumn{3}{l}{assg(x,y,z,gt)\lolli}\\
 &&\exists z_1.( \forall o,n,g. (\pr{x,o,\nonec,\nogroup}\lolli \one\ot \syncp{z_1}  \with\\
 && \pr{x,o,\unqc,\nogroup}\ot \countp{o,s(n)} \lolli \countp{o,n}\ot\syncp{z_1}  \with \\
&& \pr{x,o,\shrc,g} \ot \countp{o,s(n)} \lolli \countp{o,n}\ot\syncp{z_1}  \with \\ 
&& \pr{x,o,\immc,\nogroup} \ot \countp{o,s(n)} \lolli \countp{o,n}\ot\syncp{z_1})  \ot \\\\
&\syncp{z_1} \lolli&\exists z_2.(\pr{y,\nilp,\nonec,\nogroup} \lolli \\
&&\ \  \pr{x,\nilp,\nonec,\nogroup} \ot \pr{y,\nilp,\nonec,\nogroup} \ot \syncp{z_2}  \\
&&\with\  \forall o,n . (\pr{y,o,\unqc,\nogroup} \ot \countp{o,s(\zeroc)}\lolli\\
&&\tabs  \pr{y,o,\shrc,gt} \ot  \pr{x,o,\shrc,gt}\ot\countp{o,s(s(\zeroc))} \ot \syncp{z_2}) \\
&&\with\ \forall o,n. (\pr{y,o,\shrc,gt} \ot \countp{o,n}\lolli\\
&& \tabs \pr{y,o,\shrc,gt}\ot  \pr{x,o,\shrc,gt}\ot\countp{o,s(n)} \ot \syncp{z_2}) \\
&& \with\ \forall o,n. (\pr{y,o,\immc,\nogroup} \ot \countp{o,n}\lolli \\
&& \tabs \pr{y,o,\immc,\nogroup}\ot  \pr{x,o,\immc,\nogroup}\ot\countp{o,s(n)} \ot \syncp{z_2})\ot\\
&\syncp{z_2} \lolli&\exists z_3.(\callingpn(z)\lolli \runningpn(z)\ot\syncp{z_3}\ot\\
&\syncp{z_3}\lolli &\exists z_4.(\runningpn(z)\lolli \syncp{z_4}\otimes \bang\pe{z})))) .
\end{array}}
$
}

\resizebox{.75\textwidth}{!}{
$
\begin{array}{ll}
\cL\os P\cs =&
\!\!\!\exists c,s,svar,z,z_1,z_2,z_3,z_4,z_5.(
\pr{c,\nilp,\nonec,\nogroup}\ot\\
& \pr{s,\nilp,\nonec,\nogroup}\ot\pr{svar,\nilp,\nonec,\nogroup}\ot\syncp{z_1} \ot \bang\pe{z_1}\ot  \\
& \tabs \syncp{z_1} \lolli collection\_collection(c,z_2) \ot \\
& \tabs \syncp{z_2} \lolli stats\_stats(s,z_3) \ot \\
  & \tabs\syncp{z_3} \lolli assig(svar,s,z_4)\ot\\
  & \tabs \syncp{z_4} \lolli collection\_compStats(c,s,z_5) \ot\\
  & \syncp{z_5}\lolli  \syncp{z})\otimes (\bigotimes\limits_{i\in1..5}\pe{z_i}) \lolli \bang \pe{z} 
\end{array}
$
}

\caption{Encoding of $\texttt{assg}$ definition 
and the $\texttt{main}$ body in Example \ref{ex:deadlock}\label{fig:formula-assg}}
\end{figure}

\begin{example}[Agents as formulas]\label{ex:lcc-2-ll}
Consider the following \lcc\ process definition resulting from the encoding of the constructor of class collection in Figure \ref{fig:ae-code}:

\resizebox{\textwidth}{!}{
$\red{
\begin{array}{lll}
collection\_collection(x, z) &\defsymboldelta& 
 \exists x', o_{new} (1; 
 \pr{x,nil,\nonep,\nogroup}\to \pr{x',o_{new}, \unqc,\nogroup} \otimes \countp{o_{new},s(\zero)} ; 
 1 ; 
 \syncp{z}; \\
 & & \qquad
 \callingp{z} \to \runningp{z} ;
  \exists z'(\syncp{z'}\otimes \bang \donep{z'} \parallel (\syncp{z'} \otimes \donep{z'}) \to \\
  && \qquad\qquad
  \forall o'(\pr{x',o',\unqc,\nogroup} \otimes \countp{o' , s(\zero)} \to 
  \pr{x,o',\unqc,\nogroup} \otimes \countp{o' , s(\zero)}
  ) ; \\
 &&  \qquad\qquad\qquad\runningp{z} \to \bang \donep{z} ))
\end{array}}
$}

\noindent
\red{where the first $1$ corresponds to the  empty parallel composition in  $\texttt{gparam-init}$. From now on, for the sake of readability, we will identify $A\equiv A\ot \one$.}
This process definition gives rise to the following (universally quantified) definition
clause:

\resizebox{\textwidth}{!}{
$
\red{
\begin{array}{lll}
\multicolumn{3}{l}{
\underbrace{\!\!collection\_collection(x, z)\lolli}_{1} \exists x', o_ {new}, w_1.(\underbrace{\pr{x,nil,\nonec,\nogroup}\lolli }_{2}
(\pr{x'\!,\!o_{new} , \unqc,\nogroup} \ot
\countp{o_{new},s(\zeroc)} \ot  \syncp{w_1} \ot}\\ 
&\underbrace{\syncp{w_1}\lolli}_{3}&\exists w_2.\syncp{z}  \ot \syncp{w_2} \ot\\
 &\underbrace{\syncp{w_2}\lolli}_{4}&\exists w_3.\underbrace{\callingpn(z)\lolli}_{5} (\runningpn(z)\ot\syncp{w_3})\ot\\
   &  \ \ \ \ \ \underbrace{\syncp{w_3} \lolli}_{6}&  \exists z',w_4.\syncp{z'}\ot \bang\pe{z'}\otimes  \underbrace{ (\syncp{z'}\ot \pe{z'}\lolli }_{7}\\
&& \forall o'.(\pr{x',\!o',\!\unqc,\nogroup}\ot \countp{o',s(\zeroc)} \lolli \pr{x,o',\unqc,\nogroup} \ot  \countp{o',s(\zeroc)}  \ot\syncp{w_4}))  \ot \\
 & \ \ \ \ \ \ \underbrace{\syncp{w_4} \lolli}_{8}& \runningpn(z)\lolli \bang\pe{z}  ) ).
\end{array}}
$
}

The underlying brackets will be used in Section \ref{complexity} for determining the complexity of decomposing this formula 
The theory $\Delta$ contains the 
definition clause above  and the definition clauses for the other methods and constructors in Figure \ref{fig:ae-code} (i.e., $collection\_sort$, $collection\_print$, etc). $\Delta$ also contains the 
axioms for upgrading and downgrading permissions and the
definition clause resulting from the process definition $\texttt{assg}$ in Figure \ref{ref:rule-assg}. 
In Figure  \ref{fig:formula-assg} we show the encoding for $\texttt{assg}$ as well as the encoding $\cL\os P\cs$ of the main program in 
 Example \ref{ex:deadlock}.  
\end{example}

\subsection{Focusing and adequacy}\label{focusing}
In this section we   show that the translations presented in the last section are neat, in the sense that one computational step  corresponds to one focused  phase in proofs \cite{DBLP:journals/logcom/Andreoli92}. This will not only guarantee that our encodings are {\em adequate} (in the sense that logical proofs mimics {\em exactly} computations), but also it will provide an elegant
way of measuring the complexity of computations via complexity of derivations (see Section~\ref{complexity}).

The approach for this section will be intuitive. The reader interested in the formalization of focusing and various levels of adequacy between ILL and $\lcc$ can check the details in~\cite{TCS16}.

Let us start by analyzing the following two right rules in ILL (for the additive and multiplicative conjunctions): 
$$
\infer[\with_R]{\Gamma\lra F\with G}{\Gamma\lra F &\Gamma\ra G}\qquad
\infer[\otimes_R]{\Gamma_1,\Gamma_2\lra F\otimes G}{\Gamma_1\lra F &\Gamma_2\ra G}
$$
Reading these rules bottom-up, while the first copies the contexts, the second
involves a choice of which formulas should go to  left or right premises.  
Computationally, these behaviors are completely different: while the price to pay on applying $\with_R$  is just the duplication of memory  needed to store formulas in the context, in $\otimes_R$ one has to decide on how to split the context, and this has exponential cost.
These rules are very different from the proof theoretical point of view as well: the first rule turns out to be {\em invertible} in ILL, while the second is not. This implies that the  rule $\with_R$ can be applied {\em anywhere} in the proof, and this will not affect provability. On the other hand, $\otimes_R$ is not invertible and its application may involve backtracking.

The same analysis can be done to all other rules in ILL, giving rise to two disjoint classes of rules: the invertible ones, that can be applied eagerly,  $\{\top_R, 1_L, \otimes_L, \with_R, \lolli_R,\oplus_L,\exists_L,\forall_R, C\}$ and 
the non invertible ones $\{1_R, \otimes_R, \with_L, \lolli_L,\oplus_R,\exists_R,\forall_L, W, D,\prom\}$.

This separation induces a two phase 
proof construction: a {\em negative}, where {\em no
backtracking} on the selection of inference rules is necessary,
and a {\em positive}, where choices
within inference rules can lead to failures for which one may need to
backtrack.

%

An intuitive notion of focusing can be then stated as:  a proof is {\em focused} if, seen bottom-up,
it is a sequence of alternations between maximal
negative and positive phases.

Focusing is enough for assuring that the encoding presented in Section~\ref{ag-as-for} is, indeed, adequate.

\begin{theorem}[Adequacy ~\red{\cite{TCS16}}]\label{adequacy2}
Let $P$ be a process, 
$\Psi$ be a set of process definitions and $\Delta$ be a set of non-logical axioms. 
 Then, for any constraint $c$,
\red{$(\emptyset;P; 1) \redirex(X;\Gamma;d) \mbox{ with }\exists X. d \entails c$} 
iff  there is a proof of  the sequent 
$\bang\cL\os\Psi,\Delta\cs, \cL\os P \cs \lra  c\otimes\top
$
 in  focused ILL (ILLF).  Moreover, 
one focused logical phase corresponds exactly 
to one operational step.
\end{theorem}
This result, together with Theorem~\ref{adequacy1}, shows that AP can be adequately encoded in ILL in a natural way. In the present work we are more interested in using logic in order to verify properties of the computation, as clarified in the next example.


%
%
%
%
%
%
%

\begin{example}[Traces, proofs and focusing] \label{key}
\red{Let $A_1 =  a \to  b\to (a\otimes b) $, $A_2 = b\to a\to \okc$ and $P = a\otimes b\parallel A_1 \parallel A_2$. }
The operational semantics of \lcc\ dictates that 
there are two possible transitions leading to the store $\okc$. Both of such transitions 
start with the tell action  
 $a\otimes b$:
 \[\scriptsize
 \begin{array}{llll}
 \texttt{Derivation 1:}& 
 \conf{\emptyset}{P}{\one} &\redi^*&  \conf{\emptyset}{A_1 \parallel A_2}{a \otimes b} \redi^*  \conf{\emptyset}{b \to (a \otimes b)\parallel A_2}{ b}  \\
& &\redi^*& \conf{\emptyset}{(a \otimes b)\parallel A_2}{ \one} \redi^* 
 \conf{\emptyset}{A_2}{ a \otimes b}  \redi^*
 \conf{\emptyset}{\cdot}{\okc}\not\redi \\
 \texttt{Derivation 2:} & \conf{\emptyset}{P}{\one} &\redi^*&
  \conf{\emptyset}{A_1 \parallel A_2}{a \otimes b} \redi^* 
  \conf{\emptyset}{A_1 \parallel a\to \okc}{ b} \\
 & &\redi^*& 
  \conf{\emptyset}{A_1 \parallel \okc}{ \one} \redi^*
   \conf{\emptyset}{A_1}{\okc} \not\redi
 \end{array}
 \]
Each of these  transitions  corresponds exactly to
a focused proof of the sequent $\cL\os P\cs\lra \okc\otimes \top$: one focusing first on $\cL\os A_1\cs$
and the other focusing first on $\cL\os A_2\cs$.

On the other hand, 
there is also an interleaved execution of $A_1$ and $A_2$  that does not lead to the final store $\okc$:
$$\scriptsize\begin{array}{llcl}
 \texttt{Detivation 3:} & \langle \emptyset;P;\one \rangle&\redi^*&
\langle \emptyset;A_1 \parallel 
A_2 ;a\otimes b \rangle \redi^*
\langle \emptyset;b\to (a\otimes b)\parallel 
A_2; b \rangle\\\
& &\redi^*&\conf{\emptyset}{b\to (a\otimes b)\parallel a \to \okc}{\one} \not\redi
\end{array}$$
This  trace does not have any correspondent derivation in focused ILL (see~\cite{TCS16} for details). 
\end{example}
This example is a good witness of a need for Alcove's {\em
verifier}, other than just having an {\em animator}: an
animator exhibits traces of possible executions without
any pre-defined scheduling policy. One of such traces
may not lead to the expected final store (as the  $\texttt{ok}$ above).  On the other hand, the verifier would either fail (if a property 
is not provable) or succeed. In this last case, the proof produced
by the prover corresponds {\em  exactly} to a valid trace from the operational point of view. 

Let us show an example of how focusing can control executions on a sequential composition.

\begin{example}[Focusing on a Sequential Composition]
Consider the ILL interpretation of the sequential composition 
 $P;Q$:
\[
\cL\os P;Q\cs=
\exists z ((\cL\os(\cC\os P \cs_z)\cs) \otimes (\syncp{z} \lolli \cL\os Q\cs))
\]
This is a positive formula which will be on the left
side of the sequent and  $\exists$ and $\otimes$
will be decomposed in a negative phase. Once $P$ is executed, we observe the  invertible action of  adding the atom
$ \syncp{z}$ to the context. Then, 
one could change to a positive phase 
and focus on the negative formula
$\syncp{z} \lolli \cL\os Q\cs$. This positive action  
needs to be synchronized with the context,
consuming $ \syncp{z}$ in order to produce $\cL\os  Q\cs$. 
\end{example}

\red{In the following sections, we shall show that a focused ILL prover
is a complete decision procedure for reachability properties of the \lcc\ agents 
resulting from our encodings. This will be useful to verify properties of the encoded AP program. }

\subsection{Linear Logic as a Framework for Verifying AP  Properties}\label{sec:foc}

Let $P$ be an agent and $\cL\os
P\cs$ its translation into ILL,  
producing 
a formula $F$ together with a theory $\Delta$.  
In order to verify a certain property
$\mathcal{G}$, specified by an ILL formula $G$, we test if the
sequent $\bang\Delta,F \lra  G$ is provable.

First of all, observe that the fragment of ILL needed for encoding 
access permissions is given by the following 
grammar for guards/goals $G$ and processes $P$:
$$\begin{array}{lcl}
G & := & a \mid G\ot G \mid  \exists x.G\\
P & := & a \mid ! a\mid 1\mid  P\ot P\mid P\with P\mid
\forall x. G\lolli P\mid\exists x. P\mid  !\forall \widetilde{x}.(p(\widetilde{x})\lolli  P).
\end{array}
$$%
where $a$ is an atomic formula.
Observe that  guards $G$ do not consider banged formulas, i.e.,
agents are not allowed to \emph{ask} banged constraints. A simple inspection on the encoding of Section \ref{sec:encoding} shows that processes in our case indeed belong to such fragment.
We note also  that 
formulas generated from this grammar exhibit the following properties: 
\begin{enumerate}
 \item the left context in the sequent $\bang\Delta,F\lra G$ will be formed by
$P$ formulas;
\item  the right context will have only $G$ formulas;
\item implications on the left 
can only
introduce guards on the right side of a sequent. In fact,
on examining a proof 
bottom-up, decomposing the implication on the sequent $\Gamma_1,\Gamma_2,B\lolli 
C\lra D$ 
will produce the premises $\Gamma_1,C\lra D$ and $\Gamma_2\lra  B$. 
Hence it is important to guarantee 
that $B$ (a guard) is a $G$ (goal) formula.
\end{enumerate}

Finally, notice that the fragment described above is undecidable
in general, due to the presence of processes declarations~\cite{lmss}. 
\red{However, since we are considering AP programs adhering to the condition in Remark \ref{rem-rec-def}, our  base language does not 
lead to cyclic recursive definitions. In next section, we  determine an upper bound for  the complexity of proofs in
Alcove's verifier.   Therefore,  we can show  that 
provability in the resulting ILL 
translation is 
decidable (see Theorem~\ref{decidability}). }


\subsection{Complexity analysis} \label{complexity}
Note that, when searching for proofs in the focused system, the only non-deterministic step 
is the one 
choosing the {\em focus formula} in a positive phase. This
determines completely the complexity of a proof in ILLF and it justifies the next definition.

\begin{definition}[Proof Depth]
Let $\pi$ be a proof in ILLF.  The {\em depth} of $\pi$ is the
maximum number of positive phases along any
path in $\pi$ from the root.
\end{definition} 
\begin{example}[Complexity of Formulas]\label{ex:3}
Consider the formulas in Example \ref{ex:lcc-2-ll}.
\red{The  depth of decomposing
the definition clause  $collection\_collection(x, z)$  
into its literal or purely positive subformulas  is $8$. To see that, note that  focusing in such a negative formula on the left 
will produce 7 more  
nested positive phases in one of the branches of the proof: each one of these phases is signaled in the formula with an underlying bracket containing the respective number of the focused phase.
The same holds when decomposing the clauses for 
$stats\_stats(s, z)$ and $collection\_compStats(c,s,z)$. As we will see later, decomposing 
$assig(svar,s,z)$ has a fixed depth equal to 7. Hence the depth of a derivation
for decomposing the formula $F$  (the model of the $\texttt{main}$ program)   is $8+8+7+8=31$.} 
\end{example}

We will now proceed with a careful complexity analysis of decomposing
the formulas produced by the specification of AP programs. These will be placed
on the left of the sequent. This is done by counting 
the changes 
of
nested
polarities, as in the  example above. The complexity of
decomposing a process $P$ will be denoted by $comp(\cL\os
P\cs)$. 

\begin{itemize}[-]
\item {\bf Base cases.} We will start by presenting the complexity for
decomposing the different kinds of \lcc\ processes:
$$
\scriptsize
\begin{array}{lcl}
comp(\cL\os c\cs)& = & 0 \\
comp(\cL\os   p(x)\cs)& = & 
1+comp(\cL\os P\cs) \quad\mbox{if}  \quad\forall\widetilde{x}.p(\widetilde{x})\defsymboldelta P\\
comp( \cL\os P\parallel Q\cs ) & = &
comp( \cL\os P\cs)+comp( \cL \os Q\cs ) \\
comp(\cL\os\sum\limits_{i\in I} \forall \vx_i(c_i \to P_i)\cs)& = &
1+ max_{i\in I}\{comp( \cL\os P_i \cs)\}\\
comp(\cL\os \exists x(P) \cs)& = & comp( \cL\os P \cs). 
\end{array}
$$
\item {\bf Sequential composition.}
Recall that the process  $P;Q$ was defined in Figure~\ref{fig:seq-comp} with the aid of the function $\cC\os \cdot\cs$.  The complexity of decomposing $\cL\os P;Q\cs$ 
will be given with the help of the auxiliary 
function $comp_{sc}$, that differs from $comp$ only in the
case of the parallel composition:
$$
\scriptsize
\begin{array}{lll}
comp(\cL\os P;Q \cs )&=& 1+ comp_{sc}(\cL\os\cC\os P\cs_z\cs ) +comp(\cL\os Q\cs)\\
comp_{sc}(\cL\os\cC\os P_1 \parallel \ldots\parallel P_n\cs_z \cs)&=& 1+ \sum\limits_{i\in 1..n}comp_{sc}(\cL\os \cC\os P_i \cs_{w_i}\cs)  \\
comp_{sc}(\cC\os P\cs_{z})&=&comp(\cL\os\cC\os P\cs_{z}\cs)\quad
\mbox{in any other case}
\end{array}
$$

In the definition of
$P;Q$,  the constraint
$\syncp{z}$ will always be produced before executing $Q$. 
As already said, these are negative actions and hence do not interfere with the
proof's complexity. However,  if $P$ is  
a parallel composition $P=P_1\parallel \ldots\parallel,P_n$,
then each process $P_i$ will produce  its own
synchronization token, and  all of them will be consumed 
at once in order
to produce the constraint $\syncp{z}$. Hence, the complexity
of decomposing $P;Q$  takes into account 
nested parallel compositions inside $P$. \\

\item {\bf Wrap.} The complexity of decomposing the sub formula
$\wrapp{P,\{g_1,\cdots,g_n\},z}$ is 

\red{
\noindent
\resizebox{.95\textwidth}{!}{
$\begin{array}{lcl}comp(\cL\os\wrapp{P,\{g_1,\cdots,g_n\},z}\cs)&=&
1+comp(\bigotimes\limits_{i\in 1..n} \dgp{g_i,\concc,z} \lolli \true )\\
&&+ comp ( \cL\os P\cs \otimes ( \pe{z} \lolli \bigotimes\limits_{i\in 1..n} \dgp{g_i,\concc,z})) +1\\
&=&n+3+comp(\cL\os P\cs)
\end{array}$
}}

\item {\bf Assignment.}
It is immediate to see that: 
\[
\red{
\begin{array}{lcl}
 comp(\cL\os\gainp(x,y,gt)\cs)
&=& comp(\cL\os\dropp(x)\cs)=1
\\
comp(\cL\os\texttt{assg}(x,y,z,gt)\cs)&=&7
\end{array}}
\]
\red{Hence,
$comp(assign)=comp( \cL\os\wrapp{\texttt{assg}(x,y,z,gt),G,z}\cs)=
7+n+3=n+10
$ 
where $n$ is the number of elements in $G$. Observe that,
when there are no group permissions, the wrap is not
necessary and the complexity is the same as for decomposing
$\cL\os\texttt{assg}(x,y,z,gt)\cs$, which is 7.}\\

\item {\bf Axioms.} The upgrade and downgrade axioms  are negative  formulas. Decomposing them has depth 1.\\

\item {\bf Method definition.} 
Let $m$ 
be the number of parameters of a method and
suppose that, when consuming access permissions, 
one has to {\em upgrade} or {\em downgrade} 
$r$
 of them.
Then, $comp(\consumep)=r+(m+1)+1$, and
\red{
$$
comp(\cL\os {c\_m}(x,\widetilde{y},z)\defsymboldelta P_M\cs)
= r+ m+4+comp(\cL\os Body\cs)
$$}
where $Body$ is
the body of the method (see Rule $\rm R_{MDEF}$). 
On the other hand,
\red{
$$
comp(\cL\os Body\cs)
= comp(\cL \os \cS\os  \widehat{s}\cs _{z} \cs)+m+r+4
$$}
\item {\bf Constructor.}
With $r$ and $m$ 
as before, we have
$$
comp(\cL\os c(x,\widetilde{y},z,g_1,...,g_k)   \defsymboldelta P_C\cs)=2r+2m+comp(\cL \os \cS\os  \widehat{s}\cs _{z} \cs)+11
$$
\end{itemize} 
\begin{theorem}[Complexity]\label{decidability}
Let $\Delta$ be a theory containing the definition clauses for method and constructor definitions, the definition of $\texttt{assg}$ and the upgrade and downgrade axioms. Let  
$F$ be the  formula interpreting the main program
and $G$ be a formula interpreting a property to be proven.  
It is decidable whether or not the sequent
$\bang\Delta,F\lra G$ is provable. 
In fact, if
such a sequent is provable, then its proof
is bounded in ILLF by the depth $comp(F)+1$.
\end{theorem}
\begin{proof} 
\red{First of all, note that, since there are no circular recursive definitions (see Remark~\ref{rem-rec-def}),  methods are simply unfolded.
Moreover, as carefully described above, the complexity of such method calls is taken into account in the complexity of 
the outer method definition (see $\cL \os \cS\os  \widehat{s}\cs _{z} \cs$). This means that, whenever a method, constructor or an axiom is called in $\Delta$ via $F$, its complexity is already computed in the complexity procedure we
have just described.}
Due to the focusing discipline, proving a sequent in AP is equivalent to decomposing its formulas
completely. Therefore,  the complexity of 
the proof of the sequent $\bang\Delta,F\lra G$ is 
completely determined
by the complexity of decomposing $F$ plus the final focusing
in $G$, which is a purely positive formula.
\end{proof}

\subsection{Alcove Prover and Verification of Properties}
In the following, we explain our verification technique for three different kind of properties: deadlock detection; 
the ability of methods to run concurrently; and 
 correctness (whether programs adhere to their specifications or not). 
\red{ Recall, from Example \ref{ex:parallel}, that we have added 
to the predicates $\callingp{\cdot}, \runningp{\cdot}, \donep{\cdot}$
extra parameters to signalize the variable that called the method, the name of the method and the number of line of the source program. Then, for instance, 
in Example \ref{ex:lcc-2-ll}, the definition of the constructor looks like
}
\\

\red{
\resizebox{\textwidth}{!}{
$
\begin{array}{lll}
collection\_collection(x, z, l) &\defsymboldelta& 
 \exists x ...  \callingp{x,'collection\_collection',l,z} \to \runningp{x,'collection\_collection',l,z} ;  \\
 & &  .... \runningp{x,'collection\_collection',l,z} \to \bang \donep{x,'collection\_collection',l,z} 
\end{array}
$}
}

\red{and the encoding of, e.g, line 10 in Figure \ref{fig:ae-code}, is  $collection\_collection(x, z, 10)$. 
}
%

\noindent{\bf Deadlock Detection}.
Consider Example~\ref{ex:deadlock}. We already showed that this code
leads to a deadlock since   $s$ cannot upgrade its unique permission to
execute  \lstinline{c.compStats(s)}. We are then interested in providing a proof to
the programmer showing that the code leads to a deadlock. For doing
this, let $\cD \os De{f}\cs$ be the process definitions for 
the methods and constructors of the example  plus the definition of assignment. Let  $st$ be  the $\texttt{main}$ 
program  
 and consider the $\lcc$ program 
\red{$\cD\os De{f}\cs . \cS\os st \cs_z $ }. 
According to the definition of $\cS\os\cdot\cs$ and $\cD \os \cdot \cs$, 
we know that, for some $z$ and $c$,  $\donep{c,'collection\_compStats', 15,z}$ will be added to the store only when the statement 
\lstinline{ c.compStats(svar)} (in line 15)   is  successfully executed.
The translation of this program  will give rise to the theory
$\Delta$ and the formula $F$ described in 
Example~\ref{ex:lcc-2-ll}. 
The verification technique consists in showing that the sequent
\red{$\bang\Delta, F\lra \exists z,c. {\donep{c,'collection\_compStats',15,z}} \otimes \top$} is not provable.
This verification is done automatically by using Alcove-Prover, a  theorem prover for
ILLF developed in Teyjus and integrated to the tool described
in Section \ref{sec:tool}.  Basically, we look for proofs with depth
less or equal to \red{38}, 
given by the depth of $F$.  
In this case, the  prover  fails, thus showing that  \red{ the process $\cS\os st\cs_z$ cannot reach a store  entailing the constraint   $\exists z,c. {\donep{c,'collection\_compStats',15,z}}$.  }

The URL of the \thetool\ tool includes the output of the theorem prover and the \lcc\ interpreter for this example. It is worth noticing that the \lcc\ interpreter only computes a possible trace of the program while the theorem prover is able to check \red{all the reachable configurations for the same program}. The Alcove prover is completely faithful to the
ILLF fragment  presented in Section~\ref{sec:foc}.

The use of ``animators''  and provers is complementary. Existing formal models for system construction, such as the \emph{Rodin} (\cite{Abrial:2010:ROT:1895404.1895406}) tool for the event B modeling language, usually include both. The idea is that by using the animator the user gain  a global understanding of the behavior of the program before attempting the proof of more precise desirable properties. This usually avoids frustrations in trying to figure out corrections of the model to discharge unproved properties.\\

\noindent {\bf Concurrency Analysis.}
Consider the following \lcc\ agents
\[
\scriptsize
\begin{array}{lll}
P & = &  
\callingp{z_1} \parallel
\callingp{z_1} \to (\runningp{z_1} \otimes \syncp{z_1}) \parallel   \runningp{z_1} \to \bang\donep{z_1}
 \\
Q & = & \syncp{z_1} \to \callingp{z_2} \parallel 
	\callingp{z_2} \to  \runningp{z_2}
	\parallel  \runningp{z_2} \to \bang\donep{z_2}
\end{array}
\]
These processes represent an abstraction of the encoding of two statements $s_1$ and $s_2$ such that $s_2$ must wait until $s_1$ 
releases the program control  by adding $\syncp{z_1}$. 
It is easy to see that from the initial configuration $\gamma = \conf{\emptyset}{ P\parallel Q}{\one}$ we  always end up in the final configuration $\gamma' = \conf{\emptyset}{ \emptyset}{ \bang\donep{z_1} \ot \bang\donep{z_2}}\rangle$ showing that both $s_1$ and $s_2$ were successfully executed. Nevertheless, 
depending on the scheduler, we may observe different intermediate configurations.  For instance, if all the processes in $P$ are first selected for execution, we shall observe the derivation:
\[
\scriptsize
\begin{array}{lll}
\gamma &\redi^* &
\conf{\emptyset}{\runningp{z_1} \otimes \syncp{z_1} \parallel \runningp{z_1} \to \bang\donep{z_1}\parallel Q}{\one} \\
&\redi^* & 
\conf{\emptyset}{\runningp{z_1} \to \bang\donep{z_1}\parallel Q}{ \runningp{z_1} \ot \syncp{z_1}}\\
&\redi^*&
\conf{\emptyset}{Q}{\bang\donep{z_1}}  \redi^*   \gamma'
\end{array}
\]
On the other side, an interleaved execution of $P$ and $Q$  may be

\[
\scriptsize
\begin{array}{lll}
\gamma &\redi^* &
\conf{\emptyset}{\runningp{z_1} \otimes \syncp{z_1} \parallel \runningp{z_1} \to \bang\donep{z_1}\parallel Q}{\one} \\ 
&\redi^*&  \conf{\emptyset}{ P' \parallel Q'}{\runningp{z_1} \otimes \callingp{z_2} } \\
&\redi^* &
\conf{\emptyset}{  P' \parallel Q''}{\runningp{z_1} \otimes \runningp{z_2}}  \redi^*  \gamma'
\end{array}
\]

where $P'=\runningp{z_1} \to \bang\donep{z_1}$,
$Q'= \callingp{z_2} \to (\runningp{z_2} \otimes \syncp{z_2} \parallel Q'')$ and 
 $Q''= \runningp{z_2} \to \bang\donep{z_2}$.
Unlike the first derivation, in the second one we were able to observe the store $\runningp{z_1} \otimes\runningp{z_2}$ representing the fact that both $s_1$ and $s_2$ were executed concurrently. 

From the point of view of  the \lcc\ interpreter, 
the two derivations above are admissible. This means that 
the fact of
not observing in a trace the concurrent execution of two statements does not imply that  they have to be sequentialized due to the AP dependencies. 

We can rely on the logical view of processes to verify whether it is possible for two statements to run concurrently. For instance, consider  the Example \ref{ex:parallel} and let $F$ be the resulting ILL formula.  The following sequent turns out to be provable:
 \[
 \begin{array}{lll}
 \bang\Delta, F \lra &\exists z_1,z_2,c,s(
 \runningp{c,'collection\_print',13, z_1} \otimes \\
 & \tabs\tabs\tabs\ \ \ \ \ \ \ \  \runningp{c,'collection\_compStats',14, z_2}
  )\ot \top
  \end{array}
\]
while the following one is not:
 \[
 \begin{array}{lll}
 \bang\Delta, F \lra &\exists z_1,z_2,c,s(
 \runningp{c,'collection\_compStats',14,z_1}
 \otimes \\
 & \tabs\tabs\tabs\ \ \  \ \ \  \ \ 
  \runningp{c,'collection\_removeDuplicates',15,z_2}
  )\ot \top
  \end{array}
\]

i.e., regardless the scheduling policy, the program will not generate a   trace where $compStats$ and $removeDuplicates$   run concurrently. \\

\noindent{\bf Verifying a Method Specification.}
Finally, assume that class $collection$ 
has a field $a$ and we  define the following method
\begin{lstlisting}[numbers=none]
mistake() unq(this)=>unq(this){
   this.a<g>:=this}
\end{lstlisting}

This method requires that the unique  permission to the caller  must be restored to the
 environment. Nevertheless, the implementation of the method splits the  unique permission into two share permissions, one for the field $a$ and another for the caller (Rule $\rm R_{ALIAS}$). Then, the axiom $upgrade_1$ cannot be used to
 recover the unique permission and the ask agent in definition
 $\restenv$ remains blocked. An analysis similar to that of deadlocks will warn the  programmer about this. 
\red{In general, what we need is to prove sequents of the shape 
$\bang \Delta, \Gamma \lra \exists c,z,l . \donep{c, 'method',l,z} \otimes \top$
where $\Gamma$ contains an  atomic formula needed to start the execution of the method (i.e., a formula of the shape $c\_method(x,\cdots)$) and also the atomic formulas guaranteeing that the method can be executed  ($\pr{x,o,\unqc,\nogroup}$,  $\countp{o,s(0)}$ for the method $mistake$). This can be done, for instance,  by 
 letting $\Gamma = \cL\os \cS\os st\cs  \cs $ where $st$ is  a dummy main program that   creates an instance of \texttt{collection} and then calls the method \emph{mistake}. In this case,  the prover answers negatively to the query  $\bang \Delta, \cL\os \cS\os st\cs  \cs \lra \exists c,z,l . \donep{c,'collection\_mistake',l,z} \otimes \top$, showing that, even satisfying the preconditions of the method \emph{mistake} , it cannot finish its execution. } 

\section{Applications}\label{sec:app}
In this section we present two compelling examples of the use of our verification techniques. One is the well-known mutual exclusion problem where two (or more) processes compete for access to a critical section. In our example there are two critical sections with exclusive access. The other models a producer and a consumer processes concurrently updating a data structure.
\subsection{Two Critical Zones Management System}\label{sec:critical-zone}
Assume the class definitions for a two critical zones management system in Figure  \ref{fig:cz-example}.
\begin{figure}
\begin{lstlisting}
class lock <g> {
  lock() none(this) => unq(this) {}
  enter(process b) unq(this), shr : g(b) => unq(this), shr : g(b){}  }
class process <g>{
  attr  lock<g> lock1,  lock<g> lock2, cs<g> cs1, cs<g> cs2
  process() none(this) => unq(this) {}  }
class cs <g>{ 
  attr  lock<g> mylock
  cs() none(this) => unq(this) {
     this.mylock := new lock<g>()}
  acq1(process b,lock l)unq(this),shr:g(b),none(l)=>shr:g(this),shr:g(b),unq(l){
        l <g>:= this.mylock 
        b.cs1 <g> := this
        this.mylock <g>:= null }
  acq2(process b,lock l)unq(this),shr:g(b),none(l)=>shr:g(this),shr:g(b),unq(l){
        l<g> := this.mylock
        b.cs2 <g>:= this
        this.mylock <g>:= null  }
  release1(lock a,process b)shr:g(this),unq(a),shr:g(b) => unq(this),none(a),shr:g(b){
        this.mylock <g>:= a 
        b.cs1<g> := null 
        a <g>:= null }
  release2(lock a, process b) shr:g(this),unq(a),shr:g(b) => unq(this),none(a),shr:g(b){
        this.mylock <g>:= a     
        b.cs2 <g>:= null 
        a <g>:= null } }
    \end{lstlisting}
  \caption{Class definitions for a two critical zones management system. \label{fig:cz-example}}
\end{figure}
There are three classes, $\texttt{lock}$ (line 1), $\texttt{process}$ (line 4) and $\texttt{cs}$ (line 7). Each critical section has a private lock managed by an object of the class $\texttt{cs}$. When a process wants to enter  the  critical section $i\in \{1,2\}$, it tries first to invoke the method $\texttt{acq}i$  (lines 11 and 15) of the $\texttt{cs}$ manager. If successful, the process obtains a lock (i.e. an object of class $\texttt{lock}$) that it uses then to enter that critical zone (lines 12 and 16). When the process wants to leave the critical zone, it invokes the method $\texttt{release}i$ (lines 19 and 23). This releases ownership of the critical section lock.

Method $\texttt{acq}i$ has three parameters: $\texttt{this}$ (i.e., the $\texttt{cs}$ manager), $\texttt{b}$ the process wanting to enter the critical zone and $\texttt{l}$, a field of $\texttt{b}$ that will hold the lock of the $\texttt{cs}$ supplied by the manager. Since $\texttt{this}$ has unique permission, only one reference to the manager object can exist for $\texttt{acq}i$ to be invoked. The body of method $\texttt{acq}i$ stores the lock in $\texttt{l}$ and a reference to the manager in field $\texttt{cs1}$ or $\texttt{cs2}$ of $\texttt{b}$, depending on whether the lock for $\texttt{cs1}$ or for $\texttt{cs2}$ is requested. Storing this reference to the manager implies that it cannot longer have unique permission, so the output permission for $\texttt{this}$ becomes shared. Moreover, $\texttt{l}$ holds now the only reference to the private lock of the manager, so its output permission becomes unique. The effect is that field $\texttt{lock1}$ or $\texttt{lock2}$ of object $\texttt{b}$ uniquely acquires the section lock.
The method $\texttt{enter}$ (line 3) requires a unique permission on the lock. This ensures that only one process has a reference to the lock at any given time when entering the critical section. The method $\texttt{release}i$ 
restores  conditions as they were before invocation to $\texttt{acq}i$, i.e. the manager regains the unique permission and stores a unique reference to its private lock. Process object fields loose the lock and the reference to the manager.

\begin{figure}
\begin{tabular}{c p{.4cm}  p{.4cm}c}
\begin{lstlisting}
main () {
 group<g>
 let cs x, cs w, 
      process y, process z in
       cs1:= new cs<g>()
       cs2 := new cs<g>()
       p1 := new process<g>()
       p2 := new process<g>()
       cs1.acq1(p1, p1.lock1)
       p1.lock1.enter(p1)
       cs2.acq2(p2, p2.lock2)
       p2.lock2.enter(p2)
       cs1.acq1(p2, p2.lock1)
       p2.lock1.enter(p2)
       cs2.acq2(p1, p1.lock2)
       p1.lock2.enter(p1) 
     end }
\end{lstlisting}
& & & 
\begin{lstlisting}
   ... // constructors
   cs1.acq1(p1, p1.lock1)
   p1.lock1.enter(p1)
   cs2.acq2(p2, p2.lock2)
   p2.lock2.enter(p2)
   cs1.release1(p1.lock1, p1)
   cs1.acq1(p2, p2.lock1)
   p2.lock1.enter(p2)
   cs2.release2(p2.lock2, p2)
   cs2.acq2(p1, p1.lock2)
   p1.lock2.enter(p1)
   cs1.release1(p2.lock1, p2)
   cs2.release2(p1.lock2, p1) 
  end }
\end{lstlisting}
\\\\
(a) Deadlock code & & &  (b) Deadlock free code
\end{tabular}

\caption{Main codes  for the critical zone management system  \label{fig:prod-con-main}}
\end{figure}

Assume now the main code in Figure \ref{fig:prod-con-main} (a) where there are two section manager objects  $\texttt{cs1}$ and $\texttt{cs2}$. There are also two processes, $\texttt{p1}$ and $\texttt{p2}$.  Consider the situation where  $\texttt{p1}$ acquires the lock from  $\texttt{cs1}$ (line 9)  and enters (line 10). Then   $\texttt{p2}$ acquires the lock from 
  $\texttt{cs2}$ and enters (line 11-12). Now, $\texttt{p2}$ tries to acquire the lock from \texttt{cs2} (line 13), but this is not possible because $\texttt{cs1}$ has no longer a unique permission and execution blocks. Alcove reports this situation: 
\begin{lstlisting}
...
calling(X_6136,cs_acq1,line_71 (Z_PAR_18116))
...
[FAIL] Token ok not found. End of the program not reached.
\end{lstlisting}


Consider now the program in Figure \ref{fig:prod-con-main}    (b) where  processes  leave the critical section before attempting to acquire another lock.
In this case, all invocations run without blockage 
 and Alcove successfully finishes the analysis:
 \begin{lstlisting}
...
ended(X_6152,cs_release1,line_113 (Z_PAR_24136))
ended(W_7153,cs_release2,line_114 (Z_PAR_25137))
ok()
3517 processes Created

[OK] Token ok found. End of the program reached.
\end{lstlisting}


\subsection{Concurrent Producer-Consumer System}
\begin{figure}
{
\begin{lstlisting}
class buffer<g>{
   buffer() none(this) => unq(this) { }
   read()  shr : g(this) => shr : g(this) { }
   write()   shr : g(this) => shr : g(this) { }
   dispose() unq(this)  => unq(this) { }}
class lock { 
    lock() none(this)  => unq(this) {} } 
class cs <g>{  
     cs() none(this) => unq(this) {  } 
     acq(lock l) shr : g(this), none(l) => shr : g(this),  unq(l) { 
            l := new lock()  } 
     release(lock l) shr:g(this), unq(l) => shr:g(this),  none(l) { 
             l<g> := null  }  }
class producerConsumer<g>  {
  attr  lock l, cs<g> c 
  producerconsumer() none(this) => unq(this) {
    this.c := new cs<g>() // Initializing the critical section
  }
  produce(buffer<g> B) shr:g(this),shr:g(B) => shr:g(this),  shr:g(B) {
      this.c.acq(this.l)  // Getting a lock on the data structure
      B.write() 
      this.c.release(this.l) //releasing the lock
       }
  consume(buffer<g> B)shr:g(this),shr:g(B)=>shr:g(this),  shr:g(B) {
      this.c.acq(this.l) // Getting a lock on the data structure
      B.read()
      this.c.release(this.l)//releasing the lock
       }
  }
main{
  group <g>
  let producerConsumer<g> PC, buffer<g> B  in {
  B  := new buffer <g>()
  PC := new producerConsumer<g>()
  // produce and consume running in parallel
  split<g>{
     PC.produce(B)
     PC.consume(B)
  } }}
 \end{lstlisting}
    }
  \caption{AP Program for a concurrent producer-consumer system \label{fig:pc-example}}
\end{figure}

 Figure \ref{fig:pc-example}  shows the class definitions for a producer-consumer system working concurrently over a buffer. Class ${\it buffer}$ (line 1) represents the data structure with  operations for reading  (line 3), writing (line 4) and removing the content of the buffer (line 5). Class $producerConsummer$ (line 14) provides methods for adding ($produce$) --line 19-- and remove ($consume$) --line 24-- elements from the data structure. Since these could be invoked concurrently, the class implements a critical section (line 8) representing access to the element  of the data structure the consumer or producer is working on. That is, producing or consuming could in principle be simultaneous over different elements of the structure. To keep the example simple, we assume a single critical section over the whole data structure.
 
 Class $producerConsumer$ defines a group $g$ for processes operating over the data structure (line 14). The group is used to  manage permissions of all processes invoking methods of the class. Since callers of $produce$ and $consume$ both have share group permissions on $g$,  they can be invoked concurrently. This can be seen in the $\texttt{main}$ program (line 30). Variable $PC$ has unique permission over the $producerConsumer$ object. This unique permission is split (line 36) into share permissions for the group to allow  producer and consumer calls to run in parallel. Note, however, that simultaneous access to the buffer is precluded by the need for each process to acquire the lock before (lines 20 and 25). 
 
 As shown, in the excerpt of the Alcove's output in  Figure \ref{fig:output-PC},
 the call to $consume$ (line 4) is done while $produce$ is  still running (line 2). Note also  that before executing $write$ (line 11), the method $produce$ has to acquire the lock on the data structure (lines 5 and 9). Similarly, the execution of $read$ (line 23) (called by the consumer  in line 20) has to wait until the lock is released by the consumer (line 17) and acquired by the producer (lines 18 and 21). 
 
 AP based languages like \aeminium~\cite{CBD:2009} provides abstractions to simplify the (concurrent) access to share objects. For instance, in the example above, we locked the buffer before executing the methods $write$ and $read$ (lines 21 and 26 in  Figure \ref{fig:pc-example}). In \aeminium, it suffices  to wrap the call to these methods into a atomic block of the form:
\begin{lstlisting}[numbers=none]
atomic<g>{
 B.write()
}
\end{lstlisting}
The  \aeminium\ runtime system guarantees that 
the execution of  $write$ on the object pointed by $B$ is isolated, i.e., other methods invoked on the same object must wait until the termination of $write$. We note that the behavior of  $\texttt{atomic}$ blocks relies completely on the runtime system. Since we are interested in the static analysis of AP programs,  we did not considered atomic blocks in the grammar of Figure \ref{fig-syntax}. Note also that what we can analyze statically is whether   methods $produce$ (line 19) and $consume$ (line 24) can acquire a share permission on the buffer $B$. 

 \begin{figure}
\begin{lstlisting}[firstnumber=1]
[...] calling(PC_660,producerconsumer_produce,line_58 (z_2347))
running(PC_660,producerconsumer_produce,line_58 (z_2347))
calling(u_117119,cs_acq,line_41 (z_par_219246))
calling(PC_660,producerconsumer_consume,line_59 (z_2448))
running(u_117119,cs_acq,line_41 (z_par_219246))
calling(inner_294344,lock_lock,line 24 (z_par_307308))
calling(inner_200289,buffer_write,line_42 (z_par_220247))
running(PC_660,producerconsumer_consume,line_59 (z_2448))
running(inner_294344,lock_lock,line 24 (z_par_307308))
calling(u_117119,cs_acq,line_46 (z_par_368395))
running(inner_200289,buffer_write,line_42 (z_par_220247))
ended(inner_294344,lock_lock,line 24 (z_par_307308))
calling(u_117119,cs_release,line_43 (z_par_221248))
ended(inner_200289,buffer_write,line_42 (z_par_220247))
running(u_117119,cs_release,line_43 (z_par_221248))
ended(u_117119,cs_acq,line_41 (z_par_219246))
ended(u_117119,cs_release,line_43 (z_par_221248))
running(u_117119,cs_acq,line_46 (z_par_368395))
calling(inner_516566,lock_lock,line 24 (z_par_529530))
calling(inner_349438,buffer_read,line_47 (z_par_369396))
running(inner_516566,lock_lock,line 24 (z_par_529530))
ended(PC_660,producerconsumer_produce,line_58 (z_2347))
running(inner_349438,buffer_read,line_47 (z_par_369396)) [...]
\end{lstlisting}
 \caption{Excerpt of Alcove's output for the Producer-Consumer Program. \label{fig:output-PC}}
 \end{figure} 

\section{Concluding Remarks}\label{sec:con}
\label{sec:conclusion}

We presented an approach based on \lcc\ for specifying and verifying programs 
annotated with access permissions.  Program statements are modeled as
\lcc\ agents that faithfully represent the statement permissions flow. The
declarative reading of \lcc\ agents as formulas in
intuitionistic linear logic
permits  verifying properties such as
deadlocks, the admissibility of parallel executions, and whether methods are
correct w.r.t. their AP specifications. Central to
our verification approach is the synchronization mechanism based on
constraints, combined with the logical interpretation of \lcc\
into the focused system ILLF. 

A good strategy for
understanding the behavior of a concurrent program is running a
simulator able to observe the evolution of 
its processes, hence having a better glance of the
 global program behavior. Then, 
a prover able to verify formally various properties can be executed. 
For this reason, we have  automated our specification and 
verification approach as the Alcove tool. Using this tool we were able, for instance,
 to
verify the  critical zone
management system and the producer-consumer system presented in Section \ref{sec:app}. The reader can find these and other examples at the
Alcove tool web-site.
The results and techniques  presented here are certainly a
novel application for \ccp, and they will open a new window for the automatic verification
of (object-oriented) concurrent programs.\\

\noindent{\bf Related and Future work.}
 \ccp-based calculi have been
extensively used  to reason about  concurrent systems
in different scenarios such as system biology, security protocols, multimedia interaction systems, 
just to name a few. The reader may find in  \cite{olarte-constraints} a survey of models and  applications of \ccp.  A work related to ours is
\cite{DBLP:conf/ppdp/JagadeesanM05}, where the authors propose  a timed-\ccp\ model for
role-based access control in distributed systems. The authors combine
constraint reasoning and temporal logic model-checking to verify 
when a resource (e.g. a directory in a file system) can be accessed. 
We should also mention the work in ~\cite{DBLP:conf/lics/Nigam12} where
linear authorization logics 
are used to specify 
access control policies that may mention the
affirmations, possessions and knowledge of principals. 
In the above mentioned works,  access policies are used to control and restrict the use of resources in a distributed environment  \red{but   they do not deal with the verification of a (concurrent) programming language. }


Languages like \aeminium~\cite{CBD:2009} and Plaid
 \cite{DBLP:conf/oopsla/SunshineNSAT11} offer a series of guarantees
 such as (1) absence of AP usage protocol violation at run time; (2)
 when a program has deterministic results and (3) whether programs are
 free of race conditions on the abstract state of the objects  \cite{Typestate:Aliased:2007,DBLP:conf/sas/Boyland03}. 
Roughly, type-checking rules generate the needed information to build the  graph of dependencies among the statements in the program. Such annotations are  then used by the runtime environment to determine the pieces of code that can be executed in parallel \cite{DBLP:journals/toplas/StorkNSMFMA14}. 
Well typed  programs are free of race condition 
 by either enforcing synchronization when accessing shared data or by correctly computing dependencies. However, well typed  programs are not  necessarily deadlock free. Hence, our developments are  complementary to those works   and provide additional reasoning techniques for  AP programs.

  Somewhat surprisingly, even though in \cite{DBLP:journals/toplas/StorkNSMFMA14} it is  mentioned that ``\emph{access permissions follow the rules of linear logic}'', the authors did not go further on  this idea.  Our
linear logic encodings can be seen as the  first logic semantics for AP. As showed in this paper, such declarative reading of AP allows to perform interesting static analyses on AP based programs. 
 
%


The constraint system we propose to model the downgrade and upgrade of
axioms was inspired by the work of \emph{fractional} permissions in
\cite{DBLP:conf/sas/Boyland03} (see also
\cite{Typestate:Aliased:2007}).  \emph{Fractional} in this setting
means that an AP can be split into several more \emph{relaxed}
permissions and then joined back to form a more \emph{restrictive}
permission. For instance, a unique permission can be split into two
share permissions of weight $k/2$. Therefore, to recover a unique
permission, it is necessary to have two $k/2$-share permissions. The
constraint system described in this paper keeps explicitly the
information about the fractions by using the predicate
$\countp{\cdot}$.

Chalice \cite{DBLP:conf/vmcai/Leino10} is a program verifier for OO
concurrent programs that uses permissions to control memory
accesses. Unlike \aeminium\ and Plaid, concurrency in Chalice is
explicitly stated by the user by means of execution threads.

 \red{The language Rust (\url{https://www.rust-lang.org/}) provides mechanisms  to avoid data races. These do not use APs but rely on types. Type $mut$ (mutable) works similarly as a $unq$ permission. A data structure defined with type $mut$ is claimed ownership by the first thread using it, so it cannot be taken concurrently by another thread. The compiler checks this statically. Type $Arc$ allows the data structure to be shared among threads, but then it cannot be a mutable structure. This is then similar to $imm$ permissions. Type $Arc$ can be combined with $mutex$ to have a mutable structure that can be shared. A lock mechanism is available for the user to control simultaneous accesses. As opposed to permissions, however, there is no upgrading/downgrading of types. Hence, AP and DGAP provide, in principle,  more flexible mechanisms to express concurrent behaviors. In \cite{ullrich16masterarbeit} a translation of Rust programs into the $Lean$ prover to verify program correctness is described.
As far as we know, however, no verification of the kind we presented here is available for Rust. }

AP annotations in concurrent-by-default OO languages can be enhanced
with the notion of \emph{typestates}
\cite{Typestate:Aliased:2007,TypeState:AtomBlock:2008}. Typestates
describe abstract states in the form of state-machines, thus defining
a usage protocol (or \emph{contract}) of objects. For instance, consider the class
\emph{File} with states $\texttt{opened}$ and $\texttt{closed}$. The
signature of the method \emph{open} can be specified as the agent 
$
\pu{\thisp} \ot \texttt{closed}(\texttt{this}) \to \pu{\thisp} \ot \texttt{opened}(\texttt{this})
$.
The general idea is to verify whether a program follows correctly the
usage protocol defined by the class. For example, calling the method
$read$ on a $\texttt{closed}$ file leads to an error. Typestates then
impose certain order in which methods can be
called.  The approach our paper defines can be 
extended to deal with typestates annotations, thus widening its
applicability.

The work in 
\cite{DBLP:conf/popl/NadenBAB12} defines more specific systems and
rules for access permissions to provide for  \emph{borrowing permissions}. This   
approach  aims at dealing more effectively
with local variable aliasing, and  with how permissions flow from the
environment to method formal parameters. Considering these systems in
Alcove amounts to refine our model of permissions in Section
\ref{sec:encoding}. Verification techniques should remain the same.


%

\paragraph{Acknowledgments}
This work has been partially supported by CNPq and CAPES/Colciencias/INRIA's project STIC AmSud. We thank the anonymous reviewers for their valuable comments on an earlier draft of this paper.
\balance

\bibliographystyle{acmtrans}



\end{document}